\documentclass{article}[12]
\usepackage{amssymb}
\usepackage{amsthm}
\usepackage{ifpdf}
\usepackage{subfig}
\usepackage{wrapfig}
\usepackage{fullpage}
\usepackage{cite}
\usepackage{url}
\usepackage{caption}
\usepackage{amsmath}
\usepackage{amsfonts}
\usepackage{comment}
\usepackage{appendix}
\usepackage[nofillcomment,noend]{algorithm2e}
\usepackage{colortbl}
\usepackage{tikz}

\ifpdf

  \usepackage[pdftex]{epsfig}
  \usepackage[pdftex]{hyperref}

\else

    \usepackage[dvips]{epsfig}
    \newcommand{\href}[2]{#2}

\fi

\newif\ifabstract
\newif\iffull

\ifabstract
	\fullfalse
\else
	\fulltrue
\fi

\newtoks\magicAppendix
\magicAppendix={}
\newtoks\magictoks
\newif\iflater
\laterfalse
\ifabstract
\long\def\later#1{\magictoks={#1}%
  \edef\magictodo{\noexpand\magicAppendix={\the\magicAppendix \par
    \the\magictoks%
  }}
  \magictodo}
\long\def\both#1{\magictoks={#1}%
  \edef\magictodo{\noexpand\magicAppendix={\the\magicAppendix \par
    \noexpand\setcounter{theorem-preserve}{\noexpand\arabic{theorem}}%
    \noexpand\setcounter{theorem}{\arabic{theorem}}%
    \noexpand\setcounter{section-preserve}{\noexpand\arabic{section}}%
    \noexpand\setcounter{section}{\arabic{section}}%
	\noexpand\let\noexpand\oldsection=\noexpand\thesection
	\noexpand\def\noexpand\thesection{\thesection}
	\noexpand\let\noexpand\oldlabel=\noexpand\label
	\noexpand\let\noexpand\label=\noexpand\blank
    \the\magictoks%
    \noexpand\setcounter{theorem}{\noexpand\arabic{theorem-preserve}}%
    \noexpand\setcounter{section}{\noexpand\arabic{section-preserve}}%
	\noexpand\let\noexpand\thesection=\noexpand\oldsection
	\noexpand\let\noexpand\label=\noexpand\oldlabel
  }}
  \magictodo
  \the\magictoks}
\else
\long\def\later#1{#1}
\long\def\both#1{#1}
\fi
\long\def\magicappendix{
	\latertrue%
	\the\magicAppendix%
}

\theoremstyle{definition}
\newtheorem{theorem}{Theorem}
\newtheorem{lemma}{Lemma}
\newtheorem*{definition}{Definition}

\newcommand{\ta}{\tilde{\alpha}}
\newcommand{\tb}{\tilde{\beta}}
\newcommand{\tg}{\tilde{\gamma}}
\newcommand{\setr}[2]{\left\{\ #1 \ \left|\ #2 \right. \ \right\}}

\newcommand{\pred}{\mathrm{pred}}

\newcommand{\Z}{\mathbb{Z}}
\newcommand{\N}{\mathbb{N}}
\newcommand{\pfunc}[3]{#1 : #2 \dashrightarrow #3 }

\newcommand{\strength}{{\rm str}}

\newcommand{\dom}{{\rm dom} \;}

\newcommand{\res}[1]{\textrm{res}(#1)}
\newcommand{\termasm}[1]{\mathcal{A}_{\Box}[\mathcal{#1}]}
\newcommand{\prodasm}[1]{\mathcal{A}[\mathcal{#1}]}

\newcommand{\calT}{\mathcal{T}}
\newcommand{\lab}{{\rm label}}

\newcommand{\REPL}{\mathsf{REPR}}
\newcommand{\frakC}{\mathfrak{C}}

\newcommand{\tlatent}{\texttt{latent} }
\newcommand{\ton}{\texttt{on} }

\begin{document}

\title{Signal Transmission Across Tile Assemblies: 3D Static Tiles Simulate Active Self-Assembly by 2D Signal-Passing Tiles}

\author{
  Tyler Fochtman%
    \thanks{Department of Computer Science and Computer Engineering, University of Arkansas,
        \protect\url{tfochtma@email.uark.edu}.
        Supported in part by National Science Foundation Grant CCF-1117672.}
\and
  Jacob Hendricks%
    \thanks{Department of Computer Science and Computer Engineering, University of Arkansas,
        \protect\url{jhendric@uark.edu}.
        Supported in part by National Science Foundation Grant CCF-1117672.}
\and
  Jennifer E. Padilla
    \thanks{Dept of Chem, New York U,
    \protect\url{jp164@nyu.edu} This author's research was supported by National Science Foundation Grant CCF-1117210.}
\and
  Matthew J. Patitz%
    \thanks{Department of Computer Science and Computer Engineering, University of Arkansas,
      \protect\url{patitz@uark.edu}.
      Supported in part by National Science Foundation Grant CCF-1117672.}
\and
  Trent A. Rogers%
    \thanks{Department of Mathematical Sciences, University of Arkansas,
        \protect\url{tar003@email.uark.edu}.
        Supported in part by National Science Foundation Grant CCF-1117672.}
}
\date{}

\maketitle

\begin{abstract}
The 2-Handed Assembly Model (2HAM) is a tile-based self-assembly model in which, typically beginning from single tiles, arbitrarily large aggregations of static tiles combine in pairs to form structures.  The Signal-passing Tile Assembly Model (STAM) is an extension of the 2HAM in which the tiles are dynamically changing components which are able to alter their binding domains as they bind together.
For our first result, we demonstrate useful techniques and transformations for converting an arbitrarily complex STAM$^+$ tile set into an STAM$^+$ tile set where every tile has a constant, low amount of complexity, in terms of the number and types of ``signals'' they can send, with a trade off in scale factor.

  Using these simplifications, we prove that for each temperature $\tau>1$ there exists a 3D tile set in the 2HAM which is intrinsically universal for the class of all 2D STAM$^+$ systems at temperature $\tau$  (where the STAM$^+$ does not make use of the STAM's power of glue deactivation and assembly breaking, as the tile components of the 2HAM are static and unable to change or break bonds). This means that there is a single tile set $U$ in the 3D 2HAM which can, for an arbitrarily complex STAM$^+$ system $S$, be configured with a single input configuration which causes $U$ to exactly simulate $S$ at a scale factor dependent upon $S$.  Furthermore, this simulation uses only two planes of the third dimension. This implies that there exists a 3D tile set at temperature $2$ in the 2HAM which is intrinsically universal for the class of all 2D STAM$^+$ systems at temperature $1$. Moreover, we show that for each temperature $\tau>1$ there exists an STAM$^+$ tile set which is intrinsically universal for the class of all 2D STAM$^+$ systems at temperature $\tau$, including the case where $\tau = 1$.

While the simulation results are of more theoretical interest, showing the power of static tiles to simulate dynamic tiles when given one extra plane in 3D, the simplification results are of more practical interest for the experimental implementation of STAM tiles, since it provides potentially useful strategies for developing powerful STAM systems while keeping the complexity of individual tiles low, thus making them easier to physically implement.
\end{abstract}

%
\section{Introduction}\label{sec:intro}

Self-assembling systems are those in which large, disorganized collections of relatively simple components autonomously, without external guidance, combine to form organized structures. Self assembly drives the formation of a vast multitude of naturally forming structures, across a wide range of sizes and complexities (from the crystalline structure of snowflakes to complex biological structures such as viruses).  Recognizing the immense power and potential of self-assembly to manufacture structures with molecular precision, researchers have been pursuing the creation and study of artificial self-assembling systems.  This research has led to the steadily increasing sophistication of both the theoretical models (from the Tile Assembly Model (TAM) \cite{Winf98}, to the 2-Handed Assembly Model (2HAM) \cite{AGKS05g,DDFIRSS07}, and many others \cite{DDFIRSS07,SRTSARE,BeckerRR06,ChandranGR09,GeoTiles}) as well as experimentally produced building blocks and systems (a mere few of which include \cite{han2013dna,ke2012three,pinheiro2011challenges,C2CC37227D,rothemund2004algorithmic,NBlocks}).  While a number of models exist for passive self-assembly, as can be seen above, research into modeling active self-assembly is just beginning \cite{Nubots, Signals}.  Unlike passive self-assembly where structures bind and remain in one state, active self-assembly allows for structures to bind and then change state.

A newly developed model, the Signal-passing Tile Assembly Model (STAM) \cite{Signals}, is based upon the 2HAM but with a powerful and important difference.  Tiles in the aTAM and 2HAM are static, unchanging building blocks which can be thought of as analogous to write-once memory, where a location can change from empty to a particular value once and then never change again.  Instead, the tiles of the STAM each have the ability to undergo some bounded number of transformations as they bind to an assembly and while they are connected.  Each transformation is initiated by the binding event of a tile's glue, and consists of some other glue on that tile being turned either ``on'' or ``off''.  By chaining together sequences of such events which propagate across the tiles of an assembly, it is possible to send ``signals'' which allow the assembly to adapt during growth.  Since the number of transitions that any glue can make is bounded, this doesn't provide for ``fully reusable'' memory, but even with the limited reuse it has been shown that the STAM is more powerful than static models such as the aTAM and 2HAM (in 2D), for instance being able to strictly self-assemble the Sierpinski triangle \cite{Signals}.  A very important feature of the STAM is its asynchronous nature, meaning that there is no timeframe during which signals are guaranteed to fully propagate, and no guaranteed ordering to the arrival of multiple signals.  Besides providing a useful theoretical framework of asynchronous behavior, the design of the STAM was carefully aligned to the physical reality of implementation by DNA tiles using cascades of strand-displacement.  Capabilities in this area are improving, and now include the linear transmission of signals, where one glue binding event can activate one other glue on a DNA tile \cite{Jennifer}.

Although the STAM is intended to provide both a powerful theoretical framework and a solid basis for representing possible physical implementations, often those two goals are at odds.  In fact, in the STAM it is possible to define tiles which have arbitrary \emph{signal complexity} in terms of the numbers of glues that a tile may have on any given side and the number of signals that each tile can initiate.  Clearly, as the signal complexity of tiles increase, the ease of making these tiles in the laboratory diminishes.  Therefore, in this paper our first set of results provide a variety of methods for simplifying the tiles in STAM systems.  Besides reducing just the general signal complexity of tiles, we also seek to reduce and/or remove certain patterns of signals which may be more difficult to build into DNA-based tiles, namely \emph{fan-out} (which occurs when a single signal must split into multiple paths and have multiple destinations), \emph{fan-in} (which occurs when multiple signals must converge and join into one path to arrive at a single glue), and \emph{mutual activation} (which occurs when both of the glues participating in a particular binding event initiate their own signals).  By trading signal complexity for tile complexity and scale factor, we show how to use some simple primitive substitutions to reduce STAM tile sets to those with much simpler tiles.  Note that while in the general STAM it is possible for signals to turn glues both ``on'' and ``off'', our results pertain only to systems which turn glues ``on'' (which we call \emph{STAM$^+$} systems).

In particular, we show that the tile set for any temperature 1 STAM$^+$ system, with tiles of arbitrary complexity, can be converted into a temperature 1 STAM$^+$ system with a tile set where no tile has greater than 2 signals and either fan-out or mutual activation are completely eliminated.   We show that any temperature 2 STAM$^+$ system can be converted into a temperature 2 STAM$^+$ system where no tile has greater than 1 signal and both fan-out and mutual activation are eliminated.  Importantly, while both conversions have a worst case scale factor of $|T^2|$, where $T$ is the tile set of the original system, and worst case tile complexity of $|T^2|$, those bounds are required for the extremely unrealistic case where \emph{every} glue is on \emph{every} edge of some tile and also sends signals to \emph{every} glue on \emph{every} side of that tile.  Converting from a more realistic tile set yields factors which are on the order of the square of the maximum signal complexity for each side of a tile, which is typically much smaller.  Further, the techniques used to reduce signal complexity and remove fan-out and mutual activation are likely to be useful in the original design of tile sets rather than just as brute force conversions of completed tile sets.

We next consider the topic of intrinsic universality, which was initially developed to aid in the study of cellular automata \cite{DelormeMOT11,DelormeMOT11a}.  The notion of intrinsic universality was designed to capture a strong notion of simulation, in which one particular automaton is capable of simulating the \emph{behavior} of any automaton within a class of automata.
Furthermore, to simulate the behavior of another automaton, the simulating automaton must evolve in such a way that a translated rescaling (rescaled not only with respect to rectangular blocks of cells, but also with respect to time) of the simulator can be mapped to a configuration of the simulated automaton. The specific rescaling depends on the simulated automaton and gives rise to a global rule such that each step of the simulated automaton's evolution is mirrored by the simulating automaton, and vice versa via the inverse of the rule.

In this way, it is said that the simulator captures the dynamics of the simulated system, acting exactly like it, modulo scaling.  This is in contrast to a computational simulation, for example when a general purpose digital computer runs a program to simulate a cellular automata while the processor's components don't actually arrange themselves as, and behave like, a grid of cellular automata.  In \cite{IUSA}, it was shown that the aTAM is intrinsically universal, which means that there is a single tile set $U$ such that, for any aTAM tile assembly system $\mathcal{T}$ (of any temperature), the tiles of $U$ can be arranged into a seed structure dependent upon $\mathcal{T}$ so that the resulting system (at temperature $2$), using only the tiles from $U$, will faithfully simulate the behaviors of $\mathcal{T}$.  In contrast, in \cite{2HAMIU} it was shown that no such tile set exists for the 2HAM since, for every temperature, there is a 2HAM system which cannot be simulated by any system operating at a lower temperature.  Thus no tile set is sufficient to simulate 2HAM systems of arbitrary temperature.

For our main result, we show that there is a 3D 2HAM tile set $U$ which is intrinsically universal (IU) for the class $\frakC$ of all STAM$^+$ systems at temperature 1 and 2.  For every $\mathcal{T} \in \frakC$, a single input supertile can be created, and using just copies of that input supertile and the tiles from $U$, at temperature 2 the resulting system will faithfully simulate $\mathcal{T}$.  Furthermore, the simulating system will use only 2 planes of the third dimension.  (The signal tile set simplification results are integral in the construction for this result, especially in allowing it to use only 2 planes.)  This result is noteworthy especially because it shows that the dynamic behavior of signal tiles (excluding glue deactivation) can be \emph{fully duplicated} by static tile systems which are allowed to ``barely'' use three dimensions.  Furthermore, for every temperature $\tau>1$ there exists a 3D 2HAM tile set which can simulate the class of all STAM$^+$ systems at temperature $\tau$.

\section{Preliminaries}\label{sec:prelims}

Here we provide definitions and terms used in this paper.  See also \cite{SignalsArxiv} for a more detailed definition of the STAM.

\subsection{Informal definition of the 2HAM}

The 2HAM \cite{AGKS05g,DDFIRSS07} is a generalization of the abstract Tile Assembly Model (aTAM) \cite{Winf98} in that it allows for two assemblies, both possibly consisting of more than one tile, to attach to each other. Since we must allow that the assemblies might require translation before they can bind, we define a \emph{supertile} to be the set of all translations of a $\tau$-stable assembly, and speak of the attachment of supertiles to each other, modeling that the assemblies attach, if possible, after appropriate translation.
We now give a brief, informal, sketch of the $d$-dimensional 2HAM, for $d \in \{2,3\}$, which is normally defined as a 2D model but which we extend to 3D as well, in the natural and intuitive way.

A \emph{tile type} is a unit square if $d=2$, and cube if $d=3$, with each side having a \emph{glue} consisting of a \emph{label} (a finite string) and \emph{strength} (a non-negative integer).   We assume a finite set $T$ of tile types, but an infinite number of copies of each tile type, each copy referred to as a \emph{tile}.
A \emph{supertile} is (the set of all translations of) a positioning of tiles on the integer lattice $\Z^d$.  Two adjacent tiles in a supertile \emph{interact} if the glues on their abutting sides are equal and have positive strength.
Each supertile induces a \emph{binding graph}, a grid graph whose vertices are tiles, with an edge between two tiles if they interact.
The supertile is \emph{$\tau$-stable} if every cut of its binding graph has strength at least $\tau$, where the weight of an edge is the strength of the glue it represents.
That is, the supertile is stable if at least energy $\tau$ is required to separate the supertile into two parts.
A 2HAM \emph{tile assembly system} (TAS) is a pair $\calT = (T,\tau)$, where $T$ is a finite tile set and $\tau$ is the \emph{temperature}, usually 1 or 2. (Note that this is considered the ``default'' type of 2HAM system, while a system can also be defined as a triple $(T,S,\tau)$, where $S$ is the \emph{initial configuration} which in the default case is just infinite copies of all tiles from $T$, but in other cases can additionally or instead consist of copies of pre-formed supertiles.)
Given a TAS $\calT=(T,\tau)$, a supertile is \emph{producible}, written as $\alpha \in \prodasm{T}$, if either it is a single tile from $T$, or it is the $\tau$-stable result of translating two producible assemblies without overlap. Note that if $d=3$, or if $d=2$ but it is explicitly mentioned that \emph{planarity} is to be preserved, it must be possible for one of the assemblies to start infinitely far from the other and by merely translating in $d$ dimensions arrive into a position such that the combination of the two is $\tau$-stable, without ever requiring overlap.  This prevents, for example, binding on the interior of a region completely enclosed by a supertile.
A supertile $\alpha$ is \emph{terminal}, written as $\alpha \in \termasm{T}$, if for every producible supertile $\beta$, $\alpha$ and $\beta$ cannot be $\tau$-stably attached.
A TAS is \emph{directed} if it has only one terminal, producible supertile.

\iffull
\subsection{Formal definition of the $d$-dimensional 2HAM}
\fi
\later{
\label{def:2ham_formal}
We now formally define the 2HAM in $d \in \{2,3\}$ dimensions.

We work in the $d$-dimensional discrete space $\Z^d$. Define the set
$U_d$ to be the set of all
\emph{unit vectors} in $\mathbb{Z}^d$ (i.e. vectors of length $1$ in $\mathbb{Z}^d$).
We also sometimes refer to these vectors by the directions north, east, south, west, up, and down ($N$, $E$, $S$, $W$, $U$, $D$).
All \emph{graphs} in this paper are undirected.
A \emph{grid graph} is a graph $G =
(V,E)$ in which $V \subseteq \Z^d$ and every edge
$\{\vec{a},\vec{b}\} \in E$ has the property that $\vec{a} - \vec{b} \in U_d$.

Intuitively, a tile type $t$ is a unit square if $d=2$ and a unit cube if $d=3$ that can be
translated, but not rotated, having a well-defined ``side
$\vec{u}$'' for each $\vec{u} \in U_d$. Each side $\vec{u}$ of $t$
has a ``glue'' with ``label'' $\textmd{label}_t(\vec{u})$--a string
over some fixed alphabet--and ``strength''
$\textmd{str}_t(\vec{u})$--a nonnegative integer--specified by its type
$t$. Two tiles $t$ and $t'$ that are placed at the points $\vec{a}$
and $\vec{a}+\vec{u}$ respectively, \emph{bind} with \emph{strength}
$\textmd{str}_t\left(\vec{u}\right)$ if and only if
$\left(\textmd{label}_t\left(\vec{u}\right),\textmd{str}_t\left(\vec{u}\right)\right)
=
\left(\textmd{label}_{t'}\left(-\vec{u}\right),\textmd{str}_{t'}\left(-\vec{u}\right)\right)$.

In the subsequent definitions, given two partial functions $f,g$, we write $f(x) = g(x)$ if~$f$ and~$g$ are both defined and equal on~$x$, or if~$f$ and~$g$ are both undefined on $x$.

Fix a finite set $T$ of tile types.
A $T$-\emph{assembly}, sometimes denoted simply as an \emph{assembly} when $T$ is clear from the context, is a partial
function $\pfunc{\alpha}{\Z^d}{T}$ defined on at least one input, with points $\vec{x}\in\Z^d$ at
which $\alpha(\vec{x})$ is undefined interpreted to be empty space,
so that $\dom \alpha$ is the set of points with tiles.
We write $|\alpha|$ to denote $|\dom \alpha|$, and we say $\alpha$ is
\emph{finite} if $|\alpha|$ is finite. For assemblies $\alpha$
and $\alpha'$, we say that $\alpha$ is a \emph{subassembly} of
$\alpha'$, and write $\alpha \sqsubseteq \alpha'$, if $\dom \alpha
\subseteq \dom \alpha'$ and $\alpha(\vec{x}) = \alpha'(\vec{x})$ for
all $x \in \dom \alpha$.

Two assemblies $\alpha$ and $\beta$ are \emph{disjoint} if $\dom \alpha \cap \dom \beta = \emptyset.$
For two assemblies $\alpha$ and $\beta$, define the \emph{union} $\alpha \cup \beta$ to be the assembly defined for all $\vec{x}\in\Z^d$ by $(\alpha \cup \beta)(\vec{x}) = \alpha(\vec{x})$ if $\alpha(\vec{x})$ is defined, and $(\alpha \cup \beta)(\vec{x}) = \beta(\vec{x})$ otherwise. Say that this union is \emph{disjoint} if $\alpha$ and $\beta$ are disjoint.

The \emph{binding graph of} an assembly $\alpha$ is the grid graph
$G_\alpha = (V, E )$, where $V =
\dom{\alpha}$, and $\{\vec{m}, \vec{n}\} \in E$ if and only if (1)
$\vec{m} - \vec{n} \in U_d$, (2)
$\lab_{\alpha(\vec{m})}\left(\vec{n} - \vec{m}\right) =
\lab_{\alpha(\vec{n})}\left(\vec{m} - \vec{n}\right)$, and (3)
$\strength_{\alpha(\vec{m})}\left(\vec{n} -\vec{m}\right) > 0$.
Given $\tau \in \mathbb{N}$, an
assembly is $\tau$-\emph{stable} (or simply \emph{stable} if $\tau$ is understood from context), if it
cannot be broken up into smaller assemblies without breaking bonds
of total strength at least $\tau$; i.e., if every cut of $G_\alpha$
has weight at least $\tau$, where the weight of an edge is the strength of the glue it represents. In contrast to the model of Wang tiling, the nonnegativity of the strength function implies that glue mismatches between adjacent assemblies do not prevent them from binding, so long as sufficient binding strength is received from the (other) adjacent sides of the tiles at which the glues match.

For assemblies $\alpha,\beta:\Z^d \dashrightarrow T$ and $\vec{u} \in \Z^d$, we write $\alpha+\vec{u}$ to denote the assembly defined for all $\vec{x}\in\Z^d$ by $(\alpha+\vec{u})(\vec{x}) = \alpha(\vec{x}-\vec{u})$, and write $\alpha \simeq \beta$ if there exists $\vec{u}$ such that $\alpha + \vec{u} = \beta$; i.e., if $\alpha$ is a translation of $\beta$. Given two assemblies $\alpha,\beta:\Z^d \dashrightarrow T$, we say $\alpha$ is a \emph{subassembly} of $\beta$, and we write $\alpha \sqsubseteq \beta$, if $S_\alpha \subseteq S_\beta$ and, for all points $p \in S_\alpha$, $\alpha(p) = \beta(p)$.
Define the \emph{supertile} of $\alpha$ to be the set $\ta = \setr{\beta}{\alpha \simeq \beta}$.
A supertile $\ta$ is \emph{$\tau$-stable} (or simply \emph{stable}) if all of the assemblies it contains are $\tau$-stable; equivalently, $\ta$ is stable if it contains a stable assembly, since translation preserves the property of stability. Note also that the notation $|\ta| \equiv |\alpha|$ is the size of the supertile (i.e., number of tiles in the supertile) is well-defined, since translation preserves cardinality (and note in particular that even though we define $\ta$ as a set, $|\ta|$ does not denote the cardinality of this set, which is always $\aleph_0$).

For two supertiles $\ta$ and $\tb$, and temperature $\tau\in\N$, define the \emph{combination} set $C^\tau_{\ta,\tb}$ to be the set of all supertiles $\tg$ such that there exist $\alpha \in \ta$ and $\beta \in \tb$ such that (1) $\alpha$ and $\beta$ are disjoint (steric protection), (2) if $d=3$, or if $d=2$ and \emph{planarity} is explicitly being required, it must be possible to form $\gamma \equiv \alpha \cup \beta$ by the translation in $d$-dimensions of $\alpha$ beginning infinitely far from $\beta$ such that, at all times, $\alpha$ and $\beta$ are disjoint (3) $\gamma \equiv \alpha \cup \beta$ is $\tau$-stable, and (4) $\gamma \in \tg$. That is, $C^\tau_{\ta,\tb}$ is the set of all $\tau$-stable supertiles that can be obtained by ``attaching'' $\ta$ to $\tb$ stably, with $|C^\tau_{\ta,\tb}| > 1$ if there is more than one position at which $\beta$ could attach stably to $\alpha$.

It is common with seeded assembly to stipulate an infinite number of copies of each tile, but our definition allows for a finite number of tiles as well. Our definition also allows for the growth of infinite assemblies and finite assemblies to be captured by a single definition, similar to the definitions of \cite{jSSADST} for seeded assembly.

Given a set of tiles $T$, define a \emph{state} $S$ of $T$ to be a multiset of supertiles, or equivalently, $S$ is a function mapping supertiles of $T$ to $\N \cup \{\infty\}$, indicating the multiplicity of each supertile in the state. We therefore write $\ta \in S$ if and only if $S(\ta) > 0$.

A \emph{(two-handed) tile assembly system} (\emph{TAS}) is an ordered triple $\mathcal{T} = (T, S, \tau)$, where $T$ is a finite set of tile types, $S$ is the \emph{initial state}, and $\tau\in\N$ is the temperature. If not stated otherwise, assume that the initial state $S$ is defined $S(\ta) = \infty$ for all supertiles $\ta$ such that $|\ta|=1$, and $S(\tb) = 0$ for all other supertiles $\tb$. That is, $S$ is the state consisting of a countably infinite number of copies of each individual tile type from $T$, and no other supertiles. In such a case we write $\calT = (T,\tau)$ to indicate that $\calT$ uses the default initial state.  For notational convenience we sometimes describe $S$ as a set of supertiles, in which case we actually mean that  $S$ is a multiset of supertiles with infinite count of each supertile. We also assume that, in general, unless stated otherwise, the count for any supertile in the initial state is infinite.

Given a TAS $\calT=(T,S,\tau)$, define an \emph{assembly sequence} of $\calT$ to be a sequence of states $\vec{S} = (S_i \mid 0 \leq i < k)$ (where $k = \infty$ if $\vec{S}$ is an infinite assembly sequence), and $S_{i+1}$ is constrained based on $S_i$ in the following way: There exist supertiles $\ta,\tb,\tg$ such that (1) $\tg \in C^\tau_{\ta,\tb}$, (2) $S_{i+1}(\tg) = S_{i}(\tg) + 1$,\footnote{with the convention that $\infty = \infty + 1 = \infty - 1$} (3) if $\ta \neq \tb$, then $S_{i+1}(\ta) = S_{i}(\ta) - 1$, $S_{i+1}(\tb) = S_{i}(\tb) - 1$, otherwise if $\ta = \tb$, then $S_{i+1}(\ta) = S_{i}(\ta) - 2$, and (4) $S_{i+1}(\tilde{\omega}) = S_{i}(\tilde{\omega})$ for all $\tilde{\omega} \not\in \{\ta,\tb,\tg\}$.
That is, $S_{i+1}$ is obtained from $S_i$ by picking two supertiles from $S_i$ that can attach to each other, and attaching them, thereby decreasing the count of the two reactant supertiles and increasing the count of the product supertile. If $S_0 = S$, we say that $\vec{S}$ is \emph{nascent}.

Given an assembly sequence $\vec{S} = (S_i \mid 0 \leq i < k)$ of $\calT=(T,S,\tau)$ and a supertile $\tg \in S_i$ for some $i$, define the \emph{predecessors} of $\tg$ in $\vec{S}$ to be the multiset $\pred_{\vec{S}}(\tg) = \{\ta,\tb\}$ if $\ta,\tb \in S_{i-1}$ and $\ta$ and $\tb$ attached to create $\tg$ at step $i$ of the assembly sequence, and define $\pred_{\vec{S}}(\tg) = \{ \tg \}$ otherwise. Define the \emph{successor} of $\tg$ in $\vec{S}$ to be $\succ_{\vec{S}}(\tg)=\ta$ if $\tg$ is one of the predecessors of $\ta$ in $\vec{S}$, and define $\succ_{\vec{S}}(\tg)=\tg$ otherwise. A sequence of supertiles $\vec{\ta} = (\ta_i \mid 0 \leq i < k)$ is a \emph{supertile assembly sequence} of $\calT$ if there is an assembly sequence $\vec{S} = (S_i \mid 0 \leq i < k)$ of $\calT$ such that, for all $1 \leq i < k$, $\succ_{\vec{S}}(\ta_{i-1}) = \ta_i$, and $\vec{\ta}$ is \emph{nascent} if $\vec{S}$ is nascent.

The \emph{result} of a supertile assembly sequence $\vec{\ta}$ is the unique supertile $\res{\vec{\ta}}$ such that there exist an assembly $\alpha \in \res{\vec{\ta}}$ and, for each $0 \leq i < k$, assemblies $\alpha_i \in \ta_i$ such that $\dom{\alpha} = \bigcup_{0 \leq i < k}{\dom{\alpha_i}}$ and, for each $0 \leq i < k$, $\alpha_i \sqsubseteq \alpha$.  For all supertiles $\ta,\tb$, we write $\ta \to_\calT \tb$ (or $\ta \to \tb$ when $\calT$ is clear from context) to denote that there is a supertile assembly sequence $\vec{\ta} = ( \ta_i \mid 0 \leq i < k )$ such that $\ta_0 = \ta$ and $\res{\vec{\ta}} = \tb$. It can be shown using the techniques of \cite{Roth01} for seeded systems that for all two-handed tile assembly systems $\calT$ supplying an infinite number of each tile type, $\to_\calT$ is a transitive, reflexive relation on supertiles of $\calT$. We write $\ta \to_\calT^1 \tb$ ($\ta \to^1 \tb$) to denote an assembly sequence of length 1 from $\ta$ to $\tb$ and $\ta \to_\calT^{\leq 1} \tb$ ($\ta \to^{\leq 1} \tb$) to denote an assembly sequence of length 1 from $\ta$ to $\tb$ if $\ta \ne \tb$ and an assembly sequence of length 0 otherwise.

A supertile $\ta$ is \emph{producible}, and we write $\ta \in \prodasm{\calT}$, if it is the result of a nascent supertile assembly sequence. A supertile $\ta$ is \emph{terminal} if, for all producible supertiles $\tb$, $C^\tau_{\ta,\tb} = \emptyset$.\footnote{Note that a supertile $\ta$ could be non-terminal in the sense that there is a producible supertile $\tb$ such that $C^\tau_{\ta,\tb} \neq \emptyset$, yet it may not be possible to produce $\ta$ and $\tb$ simultaneously if some tile types are given finite initial counts, implying that $\ta$ cannot be ``grown'' despite being non-terminal. If the count of each tile type in the initial state is $\infty$, then all producible supertiles are producible from any state, and the concept of terminal becomes synonymous with ``not able to grow'', since it would always be possible to use the abundant supply of tiles to assemble $\tb$ alongside $\ta$ and then attach them.} Define $\termasm{\calT} \subseteq \prodasm{\calT}$ to be the set of terminal and producible supertiles of $\calT$. $\calT$ is \emph{directed} (a.k.a., \emph{deterministic}, \emph{confluent}) if $|\termasm{\calT}| = 1$.
} %

\subsection{Informal description of the STAM}

In the STAM, tiles are allowed to have sets of glues on each edge (as opposed to only one glue per side as in the TAM and 2HAM).  Tiles have an initial state in which each glue is either ``$\texttt{on}$'' or ``$\texttt{latent}$'' (i.e. can be switched $\texttt{on}$ later).  Tiles also each implement a transition function which is executed upon the binding of any glue on any edge of that tile.  The transition function specifies, for each glue $g$ on a tile, a set of glues (along with the sides on which those glues are located) and an action, or \emph{signal} which is \emph{fired} by $g$'s binding, for each glue in the set.  The actions specified may be to: 1. turn the glue $\texttt{on}$ (only valid if it is currently $\texttt{latent}$), or 2. turn the glue $\texttt{off}$ (valid if it is currently $\texttt{on}$ or $\texttt{latent}$).  This means that glues can only be $\texttt{on}$ once (although may remain so for an arbitrary amount of time or permanently), either by starting in that state or being switched $\texttt{on}$ from $\texttt{latent}$ (which we call \emph{activation}), and if they are ever switched to $\texttt{off}$ (called \emph{deactivation}) then no further transitions are allowed for that glue.  This essentially provides a single ``use'' of a glue and the signal sent by its binding.  Note that turning a glue $\texttt{off}$ breaks any bond that that glue may have formed with a neighboring tile. Also, since tile edges can have multiple active glues, when tile edges with multiple glues are adjacent, it is assumed that all matching glues in the $\texttt{on}$ state bind (for a total binding strength equal to the sum of the strengths of the individually bound glues).  The transition function defined for each tile type is allowed a unique set of output actions for the binding event of each glue along its edges, meaning that the binding of any particular glue on a tile's edge can initiate a set of actions to turn an arbitrary set of the glues on the sides of the same tile either $\texttt{on}$ or $\texttt{off}$.

As the STAM is an extension of the 2HAM, binding and breaking can occur between tiles contained in pairs of arbitrarily sized supertiles.  In order to allow for physical mechanisms which implement the transition functions of tiles but are arbitrarily slower or faster than the average rates of (super)tile attachments and detachments, rather than immediately enacting the outputs of transition functions, each output action is put into a set of ``pending actions'' which includes all actions which have not yet been enacted for that glue (since it is technically possible for more than one action to have been initiated, but not yet enacted, for a particular glue). Any event can be randomly selected from the set, regardless of the order of arrival in the set, and the ordering of either selecting some action from the set or the combination of two supertiles is also completely arbitrary.  This provides fully asynchronous timing between the initiation, or firing, of signals (i.e. the execution of the transition function which puts them in the pending set) and their execution (i.e. the changing of the state of the target glue), as an arbitrary number of supertile binding events may occur before any signal is executed from the pending set, and vice versa.  %

An STAM system consists of a set of tiles and a temperature value.  To define what is producible from such a system, we use a recursive definition of producible assemblies which starts with the initial tiles and then contains any supertiles which can be formed by doing the following to any producible assembly:  1. executing any entry from the pending actions of any one glue within a tile within that supertile (and then that action is removed from the pending set), 2. binding with another supertile if they are able to form a $\tau$-stable supertile, or 3. breaking into $2$ separate supertiles along a cut whose total strength is $< \tau$.

The STAM, as formulated, is intended to provide a model based on experimentally plausible mechanisms for glue activation and deactivation. %
However, while the model allows for the placement of an arbitrary number of glues on each tile side and for each of them to signal an arbitrary number of glues on the same tile, this is currently limited in practice.  Therefore, each system can be defined to take into account a desired threshold for each of those parameters, not exceeding it for any given tile type, and so we have defined the notion of \emph{full-tile signal complexity} as the maximum number of signals on any tile in a set to capture the maximum complexity of any tile in a given set.

\begin{definition}
The \emph{full-tile signal complexity} of a given tile, $t$, is the total number of all signals sent by all glues on that tile.  The full-tile signal complexity of an entire tile set $T$ is simply the maximum full-tile signal complexity of all $t \in T$.
\end{definition}

We now define a restriction of the STAM which is used throughout this paper.

\begin{definition}
We define the \emph{STAM$^+$} to be the STAM restricted to using only glue activation, and no glue deactivation.  Similarly, we say an STAM$^+$ tile set is one which contains no defined glue deactivation transitions, and an STAM$^+$ system $\mathcal{T} = (T,\tau)$ is one in which $T$ is an STAM$^+$ tile set.
\end{definition}

As the main goal of this paper is to show that self-assembly by systems using active, signalling tiles can be simulated using the static, unchanging tiles of the 3D 2HAM, since they have no ability to break apart after forming $\tau$-stable structures, all of our results are confined to the STAM$^+$.

A detailed, technical definition of the STAM model is provided in \cite{SignalsArxiv}.

\later{
\section{Definitions for simulation}

In this section, we both informally and formally define what it means for one 2HAM or STAM TAS to ``simulate'' another 2HAM or STAM TAS.  Therefore, all tilesets may be either 2HAM or STAM tile sets, and supertiles refer to active supertiles when built from STAM tile sets.  %

\subsection{Informal definitions for simulation}

Let $\mathcal{U} = (U,S_U,\tau_U)$ be the system which is simulating the system $\mathcal{T} = (T,S_T,\tau_T)$. There must be some scale factor $c \in \mathbb{N}$ at which $\mathcal{U}$ simulates $\mathcal{T}$, and we define a \emph{representation function} $R$ which maps each $c \times c$ square (sub)assembly in $\mathcal{U}$ to a tile in $\mathcal{T}$ (or empty space if it is incomplete).  Each such $c \times c$ block is referred to as a \emph{macrotile}, since that square configuration of tiles from set $U$ represent a single tile from set $T$.  We say that $\mathcal{U}$ simulates $\mathcal{T}$ under representation function $R$ at scale $c$.

To properly simulate $\mathcal{T}$, $\mathcal{U}$ must have 1. \emph{equivalent productions}, meaning that every supertile producible in $\mathcal{T}$ can be mapped via $R$ to a supertile producible in $\mathcal{U}$, and vice versa, and 2. \emph{equivalent dynamics}, meaning that when any two supertiles $\alpha$ and $\beta$, which are producible in $\mathcal{T}$, can combine to form supertile $\gamma$, then there are supertiles producible in $\mathcal{U}$ which are equivalent to $\alpha$ and $\beta$ which can combine to form a supertile equivalent to $\gamma$, and vice versa.  Note that especially the formal definitions for equivalent dynamics include several technicalities related to the fact that multiple supertiles in $\mathcal{U}$ may map to a single supertile in $\mathcal{T}$, among other issues.  Please see \ifabstract\else Section~\ref{sec:defsSim} \fi for details.

We say that a tile set $U$ is \emph{intrinsically universal} for a class of tile assembly systems if, for every system in that class, a system can be created for which 1. $U$ is the tile set, 2. there is some initial configuration which consists of supertiles created from tiles in $U$, where those ``input'' supertiles are constructed to encode information about the system being simulated, and perhaps also singleton tiles from $U$, 3. a representation function which maps macrotiles in the simulator to tiles in the simulated system, and 4. under that representation function, the simulator has equivalent productions and equivalent dynamics to the simulated system.  Essentially, there is one tile set which can simulate any system in the class, using only custom configured input supertiles.

For a tileset $T$, let $A^T$ and $\tilde{A}^T$ denote the set of all assemblies over $T$ and all supertiles over $T$ respectively. Let $A^T_{< \infty}$ and $\tilde{A}^T_{< \infty}$ denote the set of all finite assemblies over $T$ and all finite supertiles over $T$ respectively.

In what follows, let $U$ be a $d$-dimensional tile set and let $m \in \mathbb{Z}^+$. An $m$-\emph{block assembly}, or {\em macrotile},  over tile set $U$ is a partial function $\gamma : \mathbb{Z}^d_m \dashrightarrow U$, where $\mathbb{Z}_m = \{ 0,1,\ldots m-1 \}$.  Note that the dimension of the $m$-block is implicitly defined by $U$.  Let $B^U_m$ be the set of all $m$-block assemblies over $U$. The $m$-block with no domain is said to be $\emph{empty}$.  For an arbitrary assembly $\alpha \in A^U$ and $(x_0,\ldots x_{d-1})\in\Z^d$, define $\alpha^m_{x_0,\ldots x_{d-1}}$ to be the $m$-block supertile defined by $\alpha^m_{x_0,\ldots, x_{d-1}}(i_0,\ldots, i_{d-1}) = \alpha(mx_0+i_0,\ldots, mx_{d-1}+i_{d-1})$ for $0 \leq i_0, \ldots, i_{d-1}< m$.

\subsection{Formal definitions for simulation}\label{sec:defsSim}

For some tile set $T$ of dimension $d' \le d$, where $d \in \{2,3\}$ and $d' \in \{d-1,d\}$, and a partial function $R: B^{U}_m \dashrightarrow T$, define the \emph{assembly representation function} $R^*: A^{U} \dashrightarrow A^T$ such that $R^*(\alpha) = \beta$ if and only if $\beta(x_0,...,x_{d'-1}) = R(\alpha^m_{x_0,...,x_{d-1}})$ for all $(x_0,...,x_{d-1}) \in \mathbb{Z}^{d-1}$.
Let $f: \Z^{d} \rightarrow \Z^{d'}$, where $f(x_0,\ldots,x_{d-1}) = (x_0,\ldots,x_{d-1})$ if $d=d'$ and $f(x_0,\ldots,x_{d-1}) = (x_0,\ldots,x_{d'-1},0)$ if $d' = d-1$, and undefined otherwise.
    If either $U$ or $T$ is an STAM tile set, note that the assembly representation function does not consider the pending sets $\Pi$ of any constituent tiles, and thus treats active supertiles with different pending sets but otherwise identical as identical supertiles.
    $\alpha$ is said to map \emph{cleanly} to $\beta$ under $R^*$ if
    for all non empty blocks $\alpha'^m_{x_0,\ldots, x_{d-1}}$, $(f(x_0,\ldots,x_{d-1})+f(u_0,\ldots,u_{d-1})) \in \dom \alpha$ for some $u_0,\ldots, u_{d-1} \in \{-1,0,1\}$ such that $u_0^2 + \cdots + u_{d-1}^2 \leq 1$, or if $\alpha'$ has at most one non-empty $m$-block $\alpha^m_{0, \ldots, 0}$.
In other words, $\alpha'$ may have tiles on supertile blocks representing empty space in $\alpha$, but only if that position is adjacent to a tile in $\alpha$.  We call such growth ``around the edges'' of $\alpha'$ \emph{fuzz} and thus restrict it to be adjacent to only valid macrotiles, but not diagonally adjacent (i.e.\ we do not permit \emph{diagonal fuzz}).

For a given assembly representation function $R^*$, define the \emph{supertile representation function} $\tilde{R}: \tilde{A}^{U} \dashrightarrow \mathcal{P}(A^T)$ such that $\tilde{R}(\ta) = \{R^*(\alpha) | \alpha \in \ta \}$. $\ta$ is said to \emph{map cleanly} to $\tilde{R}(\ta)$ if $\tilde{R}(\ta)\in \tilde{A}^T$ and $\alpha$ maps cleanly to $R^*(\alpha)$ for all $\alpha \in \ta$.
In the following definitions, let $\mathcal{T} = \left(T,S,\tau\right)$  be a 2HAM or STAM TAS and, for some initial configuration $S_{\mathcal{T}}$, that depends on $\mathcal{T}$, let $\mathcal{U} = \left(U,S_{\mathcal{T}},\tau'\right)$ be a 2HAM or STAM TAS, and let $R$ be an $m$-block representation function $R: B^U_m \dashrightarrow T$.

\begin{definition}\label{scott-defn:alt-equiv-prod}
We say that $\mathcal{U}$ and $\mathcal{T}$ have \emph{equivalent productions} (at scale factor $m$), and we write $\mathcal{U} \Leftrightarrow_R \mathcal{T}$ if the following conditions hold:
\begin{enumerate}
    \item \label{scott-defn:simulate:equiv_prod_a}$\left\{\tilde{R}(\ta) | \ta \in \prodasm{\mathcal{U}}\right\} = \prodasm{\mathcal{T}}$.
    \item \label{scott-defn:simulate:equiv_prod_c}$\left\{\tilde{R}(\ta) | \ta \in \termasm{\mathcal{U}}\right\} = \termasm{\mathcal{T}}$.
    \item \label{scott-defn:simulate:equiv_prod_b}For all $\ta \in \prodasm{\mathcal{U}}$, $\ta$ maps cleanly to $\tilde{R}(\ta)$
\end{enumerate}
\end{definition}

\begin{definition}\label{scott-defn:alt-equiv-dynamic-t-to-s}
We say that $\mathcal{T}$ \emph{follows} $\mathcal{U}$ (at scale factor $m$), and we write $\mathcal{T} \dashv_R \mathcal{U}$ if, for any $\ta, \tb \in \prodasm{\mathcal{U}}$ such that $\ta \rightarrow_{\mathcal{U}}^1 \tb$, $\tilde{R}(\ta) \rightarrow_\mathcal{T}^{\leq 1} \tilde{R}\left(\tb\right)$.
\end{definition}

\begin{definition}\label{scott-defn:alt-equiv-dyanmic-s-to-t-strong}
We say that $\mathcal{U}$ \emph{strongly models} $\mathcal{T}$ (at scale factor $m$), and we write $\mathcal{U} \models^+_R \mathcal{T}$ if for any $\ta$, $\tb \in \prodasm{\mathcal{T}}$ such that $\tilde{\gamma} \in C^{\tau}_{\ta , \tb}$, then for all $\ta', \tb' \in \prodasm{\mathcal{U}}$ such that $\tilde{R}(\ta')=\ta$ and $\tilde{R}\left(\tb'\right)=\tb$, it must be that there exist $\ta'', \tb'', \tilde{\gamma}' \in \prodasm{\mathcal{U}}$, such that $\ta' \rightarrow_{\mathcal{U}} \ta''$, $\tb' \rightarrow_{\mathcal{U}} \tb''$, $\tilde{R}(\ta'')=\ta$, $\tilde{R}\left(\tb''\right)=\tb$, $\tilde{R}(\tilde{\gamma}')=\tilde{\gamma}$, and $\tilde{\gamma}' \in C^{\tau'}_{\ta'', \tb''}$.
\end{definition}

Note that \emph{strongly models} is in contrast to \emph{weakly models} as defined in \cite{2HAMIU}.

\begin{definition}\label{scott-defn:simulate}
Let $\mathcal{U} \Leftrightarrow_R \mathcal{T}$ and $\mathcal{T} \dashv_R \mathcal{U}$.  $\mathcal{U}$ \emph{simulates} $\mathcal{T}$ (at scale factor $m$) if $\mathcal{U} \models_R^+ \mathcal{T}$.
\end{definition}

Throughout this paper, we use the above definition of ``simulation'', which in \cite{2HAMIU} is referred to by the term ``strong simulation'', while they use the term ``simulation'' to refer to a weaker definition which make use of ``weakly models''.  Thus, here, for simulation we use the intuitively ``stronger'' notion of simulation where the simulator must in some sense more strictly model the simulated system.

\subsection{Intrinsic universality}

Let $\REPL$ denote the set of all $m$-block (or macrotile) representation functions (i.e., $m$-block supertile representation functions for some $m\in\Z^+$).
For some $d \in \{2,3\}$, let $\frakC$ be a class of $d$-dimensional tile assembly systems, and let $U$ be a $d'$-dimensional tile set for $d' \geq d$.
Note that every element of $\frakC$, $\REPL$, and $\mathcal{A}^U_{< \infty}$ is a finite object, hence can be represented in a suitable format for computation in some formal system such as Turing machines.
We say $U$ is \emph{intrinsically universal} for $\frakC$ \emph{at temperature} $\tau' \in \Z^+$ %
if there are computable functions $\mathcal{R}:\frakC \to \REPL$ and $\mathcal{S}:\frakC \to \left(A^U_{< \infty} \rightarrow \mathbb{N} \cup \{\infty\}\right)$ such that, for each $\mathcal{T} = (T,S,\tau) \in \frakC$, there is a constant $m\in\N$ such that, letting $R = \mathcal{R}(\mathcal{T})$, $S_\mathcal{T}=S(\mathcal{T})$, and $\mathcal{U}_\mathcal{T} = (U,S_\mathcal{T},\tau')$, $\mathcal{U}_\mathcal{T}$ simulates $\mathcal{T}$ at scale $m$ and using supertile representation function $R$.
That is, $\mathcal{R}(\mathcal{T})$ outputs a representation function that interprets macrotiles (or $m$-blocks) of $\mathcal{U}_\mathcal{T}$ as assemblies of $\mathcal{T}$, and $S(\mathcal{T})$ gives the initial state used to create the necessary macrotiles from $U$ to represent $\mathcal{T}$ subject to the constraint that no macrotile in $S_{\calT}$ can be larger than a single $m \times m$ square. We say that~$U$ is \emph{intrinsically universal} for $\frakC$ if it is intrinsically universal for $\frakC$ at some temperature $\tau'\in Z^+$.

} %

\section{Transforming STAM$^+$ Systems From Arbitrary to Bounded Signal Complexity}\label{sec:transform-to-simple-stam}

In this section, we demonstrate methods for reducing the signal complexity of STAM$^+$ systems with $\tau = 1$ or $\tau > 1$ and results related to reducing signal complexity.  First, we define terms related to the complexity of STAM systems, and then state our results for signal complexity reduction.

\later{
We now define what it means for an STAM system to contain fan-out and mutual activation.
\begin{definition}
For an STAM system $\mathcal{T} = (T,\sigma, \tau)$, we say that $\mathcal{T}$ contains {\bf fan-out} iff
$\exists  t = (G, L, \delta, \Pi)\in T$, $a\in\Gamma$ and $d\in \{N, S, E, W\}$ such that $|\delta(d,a)| > 1$.
In this case, we also say that $(d,a)$ {\bf fans-out}.
\end{definition}

\begin{definition}
For an STAM system $\mathcal{T} = (T,\sigma, \tau)$, we say that $\mathcal{T}$ contains {\bf mutual activation} iff
$\exists t_1 = (G_1, L_1, \delta_1, \Pi_1), t_2 = (G_2, L_2, \delta_2, \Pi_2)\in T$, $a\in\Gamma$ and $d_1, d_2\in \{N, S, E, W\}$ such that $d_1 = -d_2$ and $|\delta_1(d_1,a)|\neq 0\text{ and }|\delta_2(d_2,a)|\neq 0$.
In this case, we also say that $(d_1,a)$ {\bf mutually activates} glues in $\delta_2(d,a)$.
\end{definition}

Intuitively, fan-out corresponds to a single glue firing multiple signals, and mutual activation occurs when the binding of a pair of glues causes signals to be fired on both of the tiles participating in the binding.

} %

\subsection{Impossibility of eliminating both fan-out and mutual activation at $\tau=1$}

\begin{wrapfigure}{r}{0.4\textwidth}
\begin{center}
\includegraphics[width=1.2in]{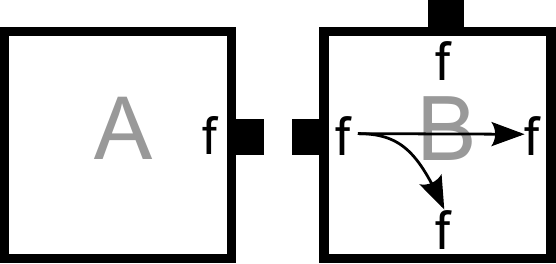}
\caption{An example of a tile set where fan-out and mutual activation cannot be completely removed. The glue $f$ on the west edge of tile type $B$ signals two other glues.}
\label{fig:noSplit}
\end{center}
\end{wrapfigure}
We now discuss the impossibility of completely eliminating both fan-out and mutual activation at temperature $1$. Consider the signal tiles in Figure~\ref{fig:noSplit} and let $\mathcal{T} = (T,1)$ be the STAM$^+$ system where $T$ consists of exactly those tiles.
Theorem~\ref{thm:fanOutEssential} shows that at temperature $1$, it is impossible to completely eliminate both fan-out and mutual activation. In other words, any STAM$^+$ simulation of $\mathcal{T}$ must contain some instance of either fan-out or mutual activation.
The intuitive idea is that the only mechanism for turning on glues is binding, and at temperature $1$ we cannot control when glues in the \texttt{on} state bind. Hence any binding pair of glues that triggers some other glue must do so by means of a sequence of glue bindings leading from the source of the signal to the signal to be turned \texttt{on}. Hence there must be paths to both of the triggered glues from the single originating glue where at some point a single binding event fires two signals. We will see that this is not the case at temperature $2$ since we can control glue binding through cooperation there.

\later{
\subsection{Proof of Lemma~\ref{lem:oneToTwo}}\label{sec:lemmaProof}

To prove Theorem~\ref{thm:fanOutEssential}, we first prove the following Lemma.

\begin{lemma}
\label{lem:oneToTwo}
At temperature $1$, any STAM$^+$ system $\mathcal{S}$ such that there exists $\alpha\in \prodasm{S}$ where a glue $g$ on some edge of $\alpha$ turns \texttt{on} two distinct glues $a$ and $b$ on edges of $\alpha$ contains
either fan-out or mutual activation.
\end{lemma}

\begin{proof}
First, we make some general observations about signal passing at temperature $1$. Notice that if a glue $g$ turns on a glue $a$, then
it must do so by a sequence of binding pairs (where a ``binding pair'' refers to the two glues that are binding together) such that the first pair in the sequence contains $g$ and the binding event of the last pair in the sequence triggers $a$. To see this, consider the fact that since $a$ is \texttt{latent} initially, a binding event must occur which initiates a signal for $a$ to turn on. The pair of glues of this binding event is the last pair in our sequence. If this pair includes $g$, the observation holds. If not, then note that the glues of this
binding pair cannot both be in the \texttt{on} state. Otherwise, assuming that $a$ does indeed get signaled to turned \ton, binding of these glues could turn $a$ \texttt{on} prior to the binding of $g$ with some other glue. Hence one of these glues must be \texttt{latent} initially and some binding event of
a pair of glues must trigger it. Now if {\em this} pair of glues includes $g$, we are done. Otherwise, we continue this argument until we reach a pair of binding glues that finally includes $g$.
This shows that the mechanism by which a glue $g$ triggers glues at temperature $1$ is a sequence of binding glues.

Another observation that we make is that if a glue $g$ turns on two distinct glues $a$ and $b$, then it must do so by means of fan-out or mutual activation.
Consider such a $g$. Then by the above argument, there exists a sequence of binding pairs such that the first pair in the sequence contains $g$ and the binding event of the last pair in the sequence triggers $a$.
Similarly, there exists such a sequence of binding pairs for $b$.
Since both of these sequences originate from the same binding pair (the one including $g$), these sequences cannot be disjoint. However, since $a$ and $b$ are distinct glues,
there must distinct glues $g_1$ and $g_2$ and a binding pair $(g_0, g_0^*)$ such that the binding of $g_0$ and $g_0^*$ fires $g_1$ and $g_2$ (where ``fires'' means sends activation signals to) without intermediate binding events.

Therefore, if a glue $g$ turns on two distinct glues $a$ and $b$ there must be at least one binding pair $(g_0, g_0^*)$ that fires $g_1$ and $g_2$ where $g_1$ and $g_2$ are distinct glues. To finish the proof we need to consider
four cases depicted in Figure~\ref{fig:mutActAndFanOutCases}. If $g_1$ and $g_2$ are glues on edges of the same tile, then we have a case of fan-out. If $g_1$ and $g_2$ are glues on separate tiles, then we have a case of mutual activation.

\begin{figure}[htp]
\begin{center}
\includegraphics[width=3in]{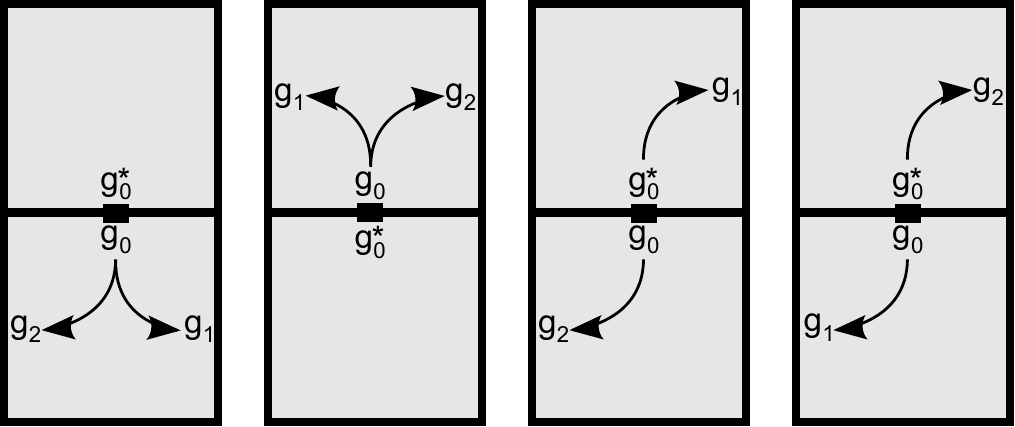}
\caption{A single binding event firing two distinct glues.  The left two use fan-out, while the right two use mutual activation.}
\label{fig:mutActAndFanOutCases}
\end{center}
\end{figure}

\end{proof}

}  %

\begin{theorem}
\label{thm:fanOutEssential}
At temperature 1, there exists an STAM$^+$ system $\mathcal{T}$ such that
any STAM$^+$ system $\mathcal{S}$ that simulates $\mathcal{T}$ contains fan-out or mutual activation.
\end{theorem}

\later{
\subsection{Proof of Theorem~\ref{thm:fanOutEssential}}\label{sec:fanOutTheoremProof}

Note that this proof makes use of Lemma~\ref{lem:oneToTwo}, which can be found in \ifabstract\else Section~\ref{sec:lemmaProof} \fi.

\begin{proof}
Let $\mathcal{S} = (S,1)$ be an STAM$^+$ system that simulates $\mathcal{T} = (T,1)$ with representation function $R$.
Given Lemma~\ref{lem:oneToTwo}, it suffices to
show that there must be an assembly $\alpha$ in $\prodasm{S}$ such that there exists a glue on an exposed edge that signals two \texttt{latent} glues to turn \texttt{on}.

Suppose that $\gamma_i^{\prime}\in \prodasm{S}$ is such that $R^*(\gamma_i^{\prime}) = \gamma_i$ where $\gamma_i \in \prodasm{T}$ is depicted in Figure~\ref{fig:pumpingB}
with $i$ tiles labeled $B$ attached. Let $t_j$ be the $j^{th}$ tile with label $B$ counting from the left in the assembly $\gamma_i$. Note that $t_0$ is of type $A$.
For each $j$, let $t_j^{\prime} \in \prodasm{S}$ be an assembly mapping to $t_j$ under $R$. Note that we can partition the glues of $t_j^{\prime}$ into four sets $\{G^j_d\}_{d\in \{N,S,E,W\}}$ where
each set corresponds to the glue $f$ on the $d$ edge of $t_j$.  To make the proof cleaner and without loss of generality, we take $j$ to be greater than $2$.

\begin{figure}[htp]
\begin{center}
\includegraphics[width=3in]{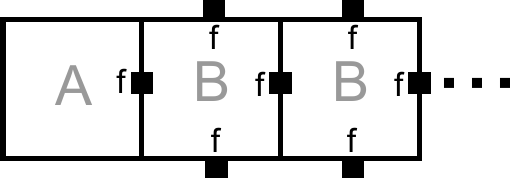}
\caption{A tile assembly in $\prodasm{T}$ consisting of $i$ repeated $B$ tiles.}
\label{fig:pumpingB}
\end{center}
\end{figure}

Since $\mathcal{T}$ follows $\mathcal{S}$, some glue of $G^j_W$ must be in the \texttt{on} state to allow $t^{\prime}_j$ to bind to glues in $G^{j-1}_E$. Hence we may assume that these same glues in $G^{j+1}_W$ are initially in the \texttt{on} state. Therefore, any matching glues of $G^j_E$ must initially be \tlatent and only turn \ton after $t^{\prime}_j$ binds to $t^{\prime}_{j-1}$ along glues in $G^{j}_W$.
Otherwise, at temperature $1$, $t^{\prime}_j$ would be able to bind to $t^{\prime}_{j+1}$ prior to $t^{\prime}_{j}$ binding to $t^{\prime}_{j-1}$, and
this would contradict the assumption that $\mathcal{T}$ follows $\mathcal{S}$. Likewise, some
glues in $G^j_S$ must initially be \tlatent and only turn \ton after $t^{\prime}_j$ binds to $t^{\prime}_{j-1}$ along glues in $G^{j}_W$.
Hence the binding of one or more glues in $G^j_W$ fires glues in $G^j_E$ and $G^j_S$.
Next we show that under the assumption that $\mathcal{S}$ simulates $\mathcal{T}$, it cannot be the case that each glue in $G^{j}_W$ fires at most one glue.

For the sake of contradiction, assume that each binding pair of glues in the production of $\gamma_i^{\prime}$
signals at most one other glue. Then the number of \tlatent glues that could be turned \ton in $G^j_E$ is at most $|G^j_W| - 1$.
This follows from the facts that signaling only occurs when glues bind and at least $1$ signal must be used to turn \ton a glue in $G^j_S$. Therefore, as $j$ increases, the number of possible glues in $G^j_W$ that could be turned on decreases. The idea is depicted in Figure~\ref{fig:losingSignal}.
\begin{figure}[htp]
\begin{center}
\includegraphics[width=3in]{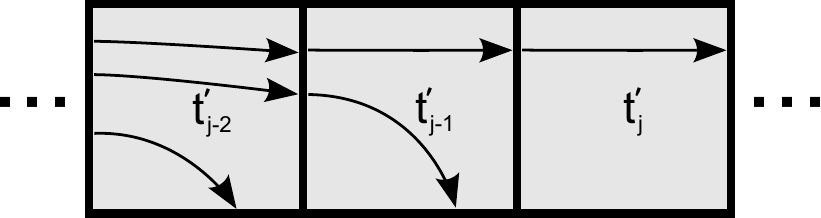}
\caption{An example of a case where there are no glues of $G^{j}_S$ that are signaled.}
\label{fig:losingSignal}
\end{center}
\end{figure}
Therefore, there exists an $n$ sufficiently large such that for some $i,j > n$ with $j < i + 1$, either no glue in $G^j_E$ is ever fired or no glue in $G^j_S$ is ever fired even after attachment of $t^{\prime}_j$.
Figure~\ref{fig:losingSignal} shows an example where no glues of $G^j_S$ are fired.
However, two signals must be sent through $t^{\prime}_{j+1}$ to turn on at least $1$ glue in $G^{j+1}_E$ and $1$ glue in $G^{j+1}_S$; therefore we have a contradiction. Hence, there must be a binding event in the production of $\gamma_i^{\prime}$
that signals two or more glues. By Lemma~\ref{lem:oneToTwo}, this implies that $\gamma_i^{\prime}$ contains fan-out or mutual activation.

\qed
\end{proof}
}  %

\subsection{Eliminating either fan-out or mutual activation}

In this section we will discuss the possibility of eliminating fan-out from an STAM$^+$ system. We do this by simulating a given STAM$^+$ system with a simplified STAM$^+$ system that contains no fan-out, but does contain mutual activation. A slight modification to the construction that we provide then shows that mutual activation can be swapped for fan-out.

\begin{definition}\label{defn:simplified-stam}
An \emph{$n$-simplified STAM tile set} is an STAM tile set which has the following properties: (1) the full-tile signal complexity is limited to a fixed constant $n \in \mathbb{N}$, (2) there is either no fan-out or no mutual activation, and (3) fan-in is limited to $2$. We say that an STAM system $\mathcal{T} = (T,\sigma,\tau)$ is $n$-simplified if $T$ is $n$-simplified.
\end{definition}

\begin{theorem}
\label{thm:2SimpAtTemp1}
For every STAM$^+$ system $\mathcal{T} = (T,\sigma, \tau)$, there exists a $2$-simplified STAM$^+$ system $\mathcal{S} = (S, \sigma^{\prime}, \tau)$ which simulates $\mathcal{T}$ with scale factor $O(|T|^2)$ and tile complexity $O(|T|^2)$.
\end{theorem}

To prove Theorem~\ref{thm:2SimpAtTemp1}, we construct a macrotile such that every pair of signal paths that run in parallel are never contained on the same tile.  This means that at most two signals are ever on one tile since it is possible for a tile to contain at most two non-parallel (i.e. crossing) signals.  In place of fan-out, we use mutual activation gadgets (see Figure~\ref{fig:echoTile}) within the {\em fan-out zone}. Similarly, we use a {\em fan-in zone} consisting of tiles that merge incoming signals two at a time, in order to reduce fan-in. For examples of these zones, see Figure~\ref{fig:constructionOverview}. Next, we print a circuit (a system of signals) around the perimeter of the macrotile which ensures that the external glues (the glues on the edges of the macrotiles that cause macrotiles to bind to one another) are not turned on until a macrotile is fully assembled. More details of the construction can be found in \ifabstract\else Section~\ref{sec:MacroCreation} \fi.
\begin{figure}[htp]
\begin{center}
\includegraphics[width=4.5in]{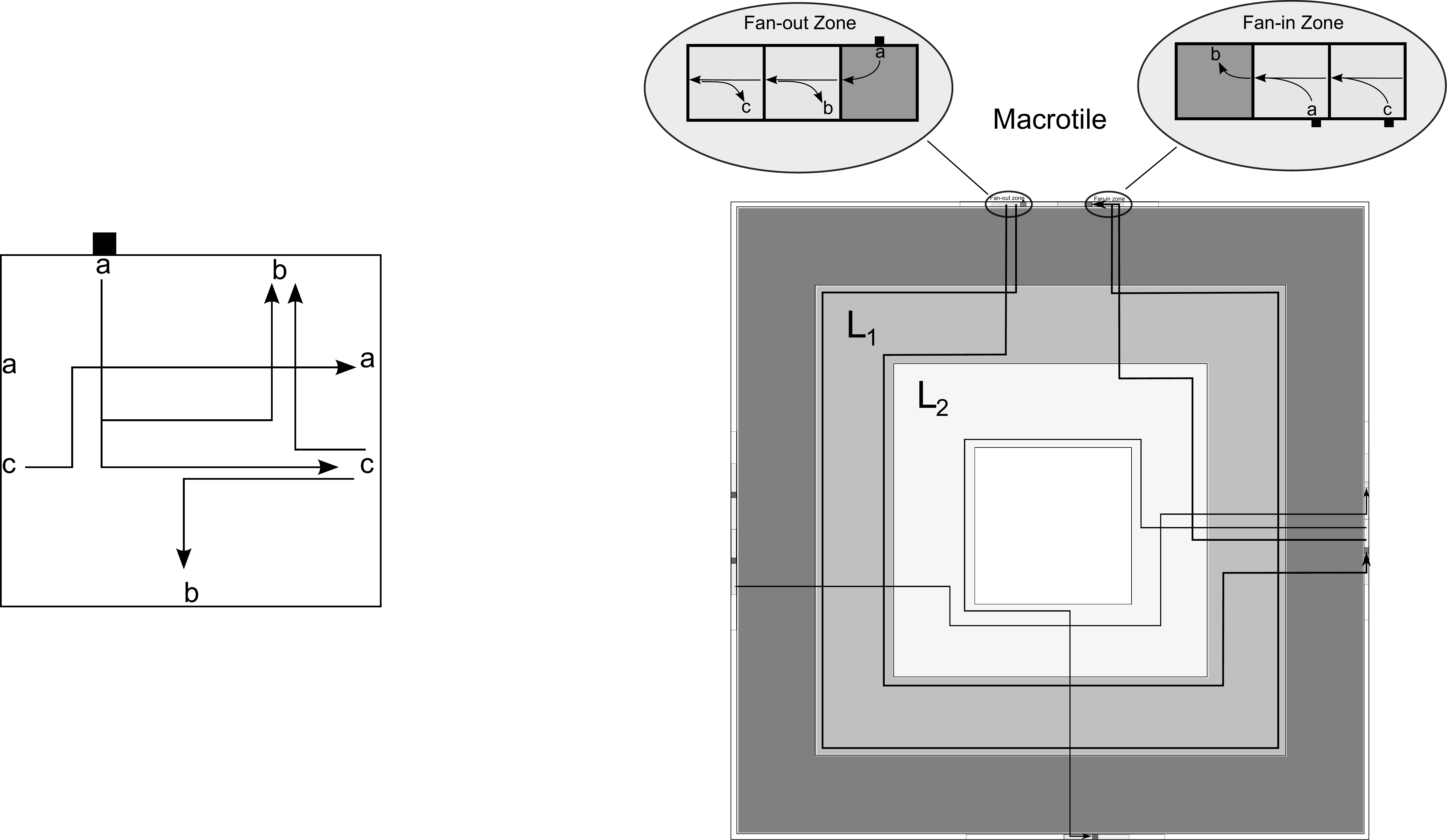}
\caption{A tile with $5$ signals (left) and the STAM$^+$ macrotile that simulates it (right).  Each frame, labeled $L_1$ and $L_2$, corresponds to a glue.  For example on the tile to be simulated (left) there is a signal that runs from glue $a$ to glue $c$.  In order to simulate this signaling, a signal runs from the fan-out zone of glue $a$ to the frame associated with glue $a$, labeled $L_1$, on the north edge.  The signal then wraps around the frame until it reaches the east side on which glue $c$ lies.  Then the signal enters the fan-in zone of glue $c$.}
\label{fig:constructionOverview}
\end{center}
\end{figure}

To further minimize the number of signals per tile at $\tau > 1$, cooperation allows us to
reduce the number of signals per tile required to just $1$. To achieve this result, we modify the construction used to show Theorem~\ref{thm:2SimpAtTemp1}, and prove Theorem~\ref{thm:1SimpAtTemp2}.
The details of the modification are in \ifabstract\else Section~\ref{sec:reduceSigPerTileTemp2} \fi.

\begin{theorem}
\label{thm:1SimpAtTemp2}
For every STAM$^+$ system $\mathcal{T} = (T,\sigma, \tau)$ with $\tau > 1$, there exists a $1$-simplified STAM$^+$ system $\mathcal{S}=(S,\sigma^{\prime},\tau)$ which simulates $\mathcal{T}$ with scale factor $O(|T|^2)$ and tile complexity $O(|T|^2)$.
\end{theorem}

\later{
\subsection{Macrotile Creation for Theorem~\ref{thm:2SimpAtTemp1}}\label{sec:MacroCreation}
In this section we provide a construction to prove the following Lemma.

\begin{lemma}\label{lem:stam-simplification}
For every STAM$^+$ system $\mathcal{T} = (T,\tau)$, there exists a $4$-simplified STAM$^+$ system $\mathcal{S} = (S, \tau)$ which simulates $\mathcal{T}$ with scale factor $O(|T|^2)$ and tile complexity $O(|T|^2)$.
\end{lemma}

The impetus behind this construction is to trade signal complexity for tile type complexity, meaning more but simpler tile types, along with an increased scale factor.  Given an arbitrary STAM$^+$ TAS $\mathcal{T} = (T, \tau)$, we show that there exists a $4$-simplified STAM$^+$ TAS $\mathcal{S} = (S, \tau)$ such that $\mathcal{S}$ simulates $\mathcal{T}$. In order to accomplish this, for each tile type $t \in T$ we hardcode a set of tiles that assemble a macrotile in $\prodasm{S}$ which exactly mimics the behavior of $t$ at scale factor $O(|T|^2)$, and tile complexity $O(|T|^2)$.  A macrotile in $\prodasm{S}$ is a square of dimensions $O(|T|)^2 \times O(|T|)^2$ which in our case can be thought of as representing a unique tile $t \in T$. Using a macrotile to simulate $t$ allows us to use mutual activation gadgets to eliminate fan-out (See Figure~\ref{fig:echoTile}), and to distribute the signal complexity of $t$ across multiple tiles.  We then have a macrotile in $\mathcal{S}$ in which all of the constituent tiles are $4$-simplified and which simulates the behavior of $t$.

} %

\begin{figure}[htp]
\begin{center}
\includegraphics[width=2.5in]{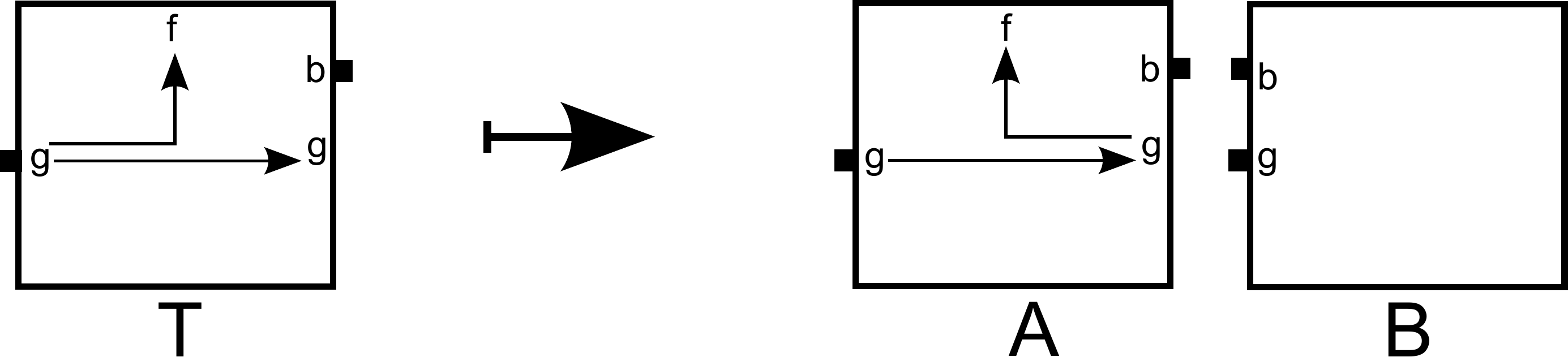}
\caption{An example of a mutual activation gadget consisting of tiles $A$ and $B$  without fan-out simulating, at $\tau=1$, the functionality of tile $T$ which has fan-out.  The glue $b$ represents the generic glues which holds the macrotile together.  The idea is to ``split'' the signals from the west glue $g$ on tile $A$ into two signals without using fan-out.  Once the west glue $g$ on tile $A$ binds, it turns on the east glue $g$ on tile $A$.  Then, when the east glue $g$ on tile $A$ binds to tile $B$, it triggers glue $f$.  Thus, the east glue $g$ triggers both the west glue $g$ and glue $f$ without fan-out.}
\label{fig:echoTile}
\end{center}
\end{figure}

\later{

This construction takes as input an STAM$^+$ TAS $\mathcal{T} = (T, \tau)$ and outputs an STAM$^+$ TAS $\mathcal{S} = (S, \tau)$ such that for any $s \in S$ there is no fan-out and the number of signals on $s$ is $\le 2$.
Notice that for a general STAM$^+$ system $\mathcal{T} = (T,\sigma,\tau)$, the initial configuration, $\sigma$, can consist of some pre-formed supertiles (including singleton tiles). For our construction, we take $\sigma$ to be an infinite number of copies of all singleton tiles. If additional supertiles are included in the initial configuration of $\mathcal{T}$, then we can simply form combinations of the necessary macrotiles in $\prodasm{S}$ to obtain an equivalent initial configuration for $\mathcal{S}$.

First, we introduce some useful notation. As usual, we let $G$ be the set of glues contained in the tile set $T$.  Let $i \in D$, where $D = \{N,E,S,W\}$, be a direction.  We say that glue $g \in G_i$ if $g$ is on edge $i$ of any $t \in T$.  Also, we define the multiset $\mathcal{G}$ as $\mathcal{G} = \cup_{i=1}^4 G_i$.

For each $t \in T$ we construct a unique macrotile in the following manner.
We begin by designing a set of tiles that assembles to build a $|\mathcal{G}|^2 \times |\mathcal{G}|^2$ square which presents $\tau$ strength glues around the perimeter, and whose interior glues are unique to each adjacent edge between tiles and of strength $\tau$.  This square serves no other purpose than to provide the necessary spacing required for the rest of the construction.  Then we design a set of tiles that forms a square frame of width $4|\mathcal{G}|$ such that the previously built square fits into the frame and all sides of the square bind to the frame.  We then design another frame that is again of width $4|\mathcal{G}|$ and encompasses and binds to all sides of the previous frame.  We continue in this way until all $g \in G$ are associated with a unique frame.  For each frame of width $|\mathcal{G}|$ contained in a frame of width $4|\mathcal{G}|$, we associate a direction $i \in D$.  In addition, for each single tile wide frame contained in a frame of width $|\mathcal{G}|$ that is associated with direction $i$, we associate a unique glue in $G_i$.  We then design these single tile wide frames such that for any signal that exists from $g \in G_i$ to $g^{\prime} \in \mathcal{G}$ on $t$ there is a tile that accepts the signal coming from the fan-out zone of glue $g$ on edge $i$ and propagates that signal to the fan-in zone of glue $g^{\prime}$ by carrying it around the frame to the edge on which $g^{\prime}$ lies.

Finally, we design a tile set that forms single tile wide lines for each edge of the macrotile.  The lines of tiles corresponding to edge $i$ will expose the active and latent glues of $t$ (all initially in the \tlatent state) that are contained in $G_i$ in a specified ordering.  This ordering is necessary because otherwise it could be the case that two tiles in $T$ are able to bind, but the two macrotiles in $U$ representing these tiles could not bind because they placed the corresponding glues that are supposed to bind on tiles in different positions.  Also, by having an ordering we prevent macrotiles from binding off center, that is macrotiles binding but not having their edges exactly line up, since we create a new unique glue for each position.  Each of the exterior glues in $G_i$ will have a \emph{fan-in zone} to their clockwise direction and a \emph{fan-out zone} to their counterclockwise direction.  The fan-in zone consists of a line of $|\mathcal{G}|$ tiles that is able to accept a signal from each $g \in \mathcal{G}$ and ``merge'' it into other signals so that no more than a fan-in of $2$ occurs in the fan-in zone.  For an example of a fan-in zone see Figure~\ref{fig:fanInZone}.
 \begin{figure}[htp]
\begin{center}
\includegraphics[width=3.0in]{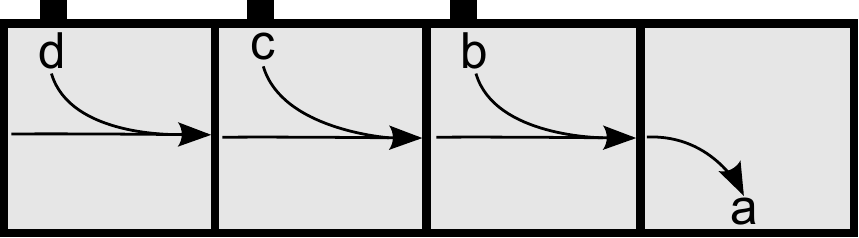}
\caption{A segment of the fan-in zone for an exterior glue $a$ on the South side of a macrotile.  Here, we see how we can use tiles to ``merge'' two signals at a time using fan-in so that at most only two binding events signal a single glue. For example, a binding event for $c$ merges with the incoming signal triggered by a glue on the West edge so that either can signal the glue on the East edge.}
\label{fig:fanInZone}
\end{center}
\end{figure}
 \noindent The fan-out zone of a glue $g$ consists of a line of $|\mathcal{G}|$ tiles which use mutual activation gadgets (Figure~\ref{fig:echoTile}) to eliminate fan-out.  An example of a fan-out zone can be seen in Figure~\ref{fig:fanOutZone}.
\begin{figure}[htp]
\begin{center}
\includegraphics[width=3.0in]{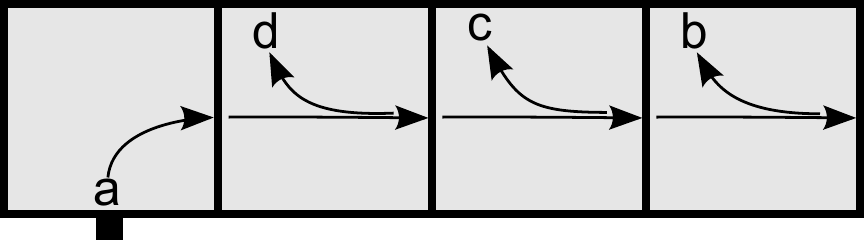}
\caption{A segment of the fan-out zone for an exterior glue $a$ on the South side of a macrotile.  Here, we see how the echo tiles are able to eliminate fan-out for glue $a$.}
\label{fig:fanOutZone}
\end{center}
\end{figure}
\noindent The fan-in and fan-out zones have glues on their ends such that they attach to the fan-in and fan-out zones of other glues to form a line of length $(2|\mathcal{G}+1)|G_i)$ with the ordering specified above. These lines of tiles are designed so that they attach to their corresponding edge and are centered with respect to the inner most $|\mathcal{G}|^2 \times |\mathcal{G}|^2$ square.  The rest of the perimeter of the macrotile is filled with generic tiles.  See Figure~\ref{fig:MacrotileOverview} for an example of a macrotile.

\begin{figure}[htp]
\begin{center}
\includegraphics[width=4.5in]{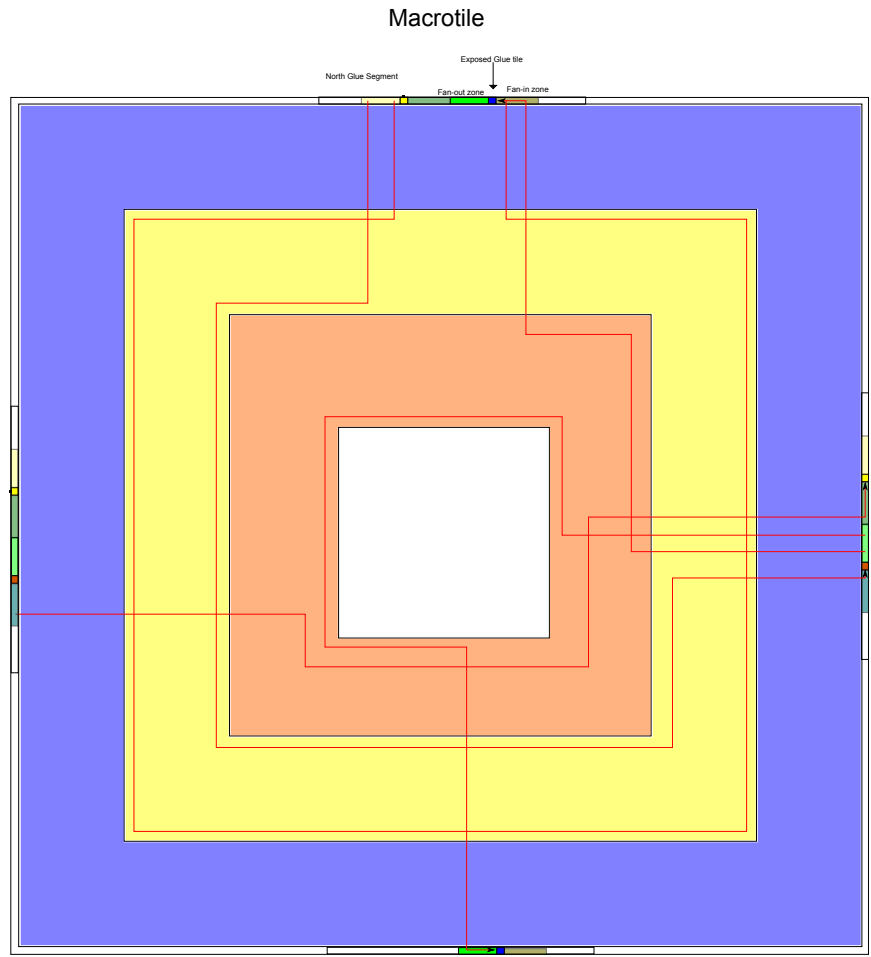}
\caption{An example of an STAM$^+$ macrotile that simulates the tile shown in Figure~\ref{fig:fanout_t_image}.  Here, the yellow squares represent glue $a$, the blue square represents glue $b$ and the orange squares represent glue $c$.  The color of each frame corresponds to the glue of the same color.  For example in Figure~\ref{fig:fanout_t_image} there is a signal that runs from glue $a$ to glue $c$.  In order to simulate this signaling, a signal runs from the fan-out zone of glue $a$ (the yellow glue) to the frame associated with glue $a$ on the north edge.  The signal then wraps around the frame until it reaches the east side on which glue $c$ lies.  Then the signal enters the fan-in zone of glue $c$.  }
\label{fig:MacrotileOverview}
\end{center}
\end{figure}

\begin{figure}[htp]
\begin{center}
\includegraphics[width=2.5in]{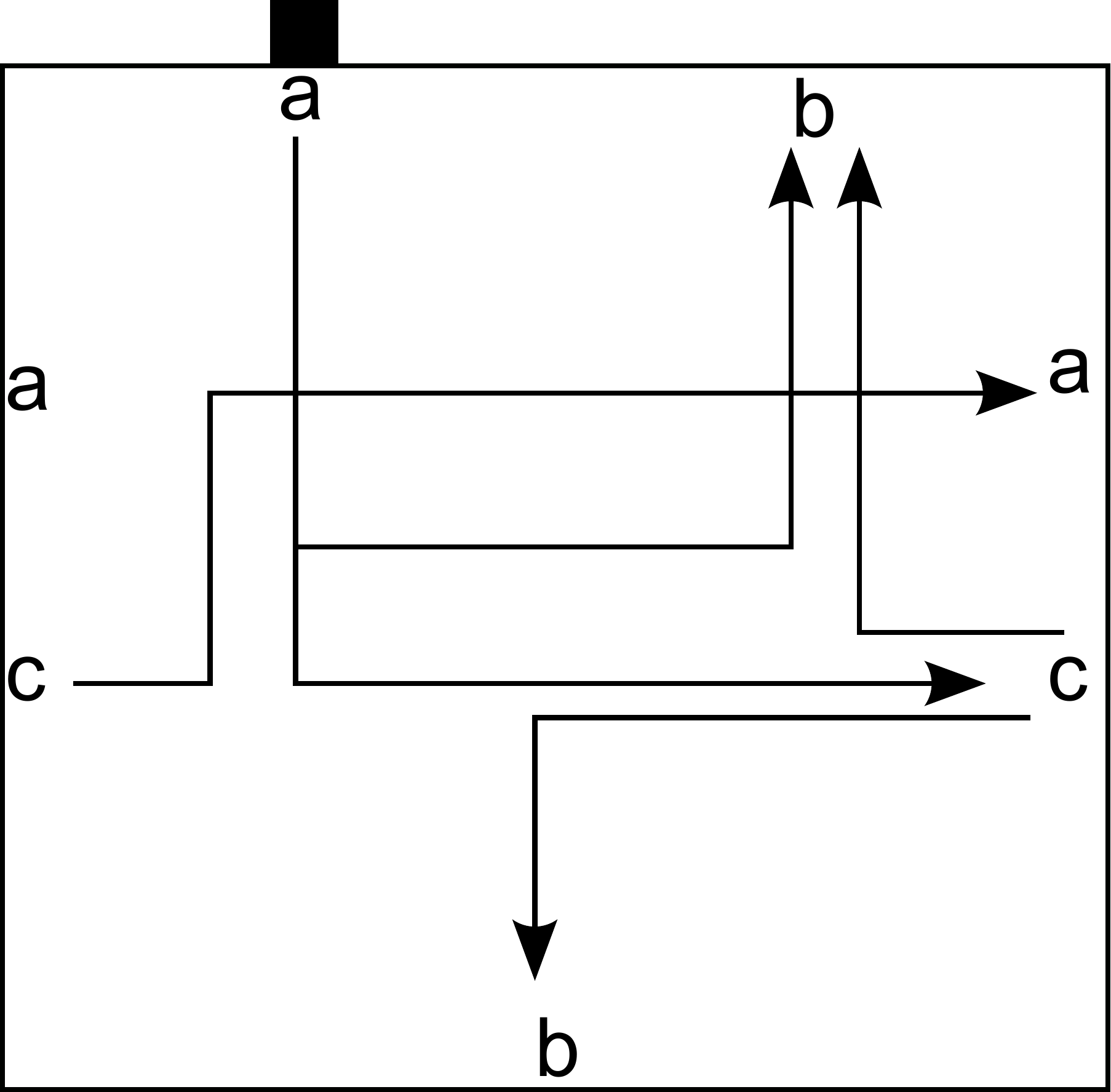}
\caption{The tile simulated by Figure~\ref{fig:MacrotileOverview}.}
\label{fig:fanout_t_image}
\end{center}
\end{figure}

\subsubsection{Macrotile Algorithm}
Let $U$ be an enumerated set of $|T|$ hardcoded macrotiles in $S$ of length $10|\mathcal{G}|^2 + 2$. We informally define the function $PrintPathToFrame$ as a function which prints the signal path from the fan-out zone of $g_k$ to the single width frame associated with the signal that travels from $g_k$ to $g_l$ (i.e. it creates the necessary series of initially \ton and \tlatent glues plus signals which combine to form a directed signalling pathway along that route).  Similarly, the function $PrintPathAroundFrame$ prints the signal path from the point where the signal path enters the frame to the tile on the frame that is closest to the merge tile associated with $g_k$ in the fan-in zone of $g_l$.  The function $PrintPathFromFrame$ prints the path from the frame to the merge tile associated with $g_k$ in the fan-in zone of $g_l$.  Using these subroutines, the algorithm for macrotile creation is given in Figure~\ref{fig:mainMacrotileCreation}.

\begin{figure}[t]
\begin{algorithm}[H]
  \KwData{$\mathcal{T}$, $U$}
  \KwResult{$\mathcal{S}$}

    \For {$t_n = (G_n, L_n, \delta_n, \Pi_n) \in T$}{
    Print hardcoded $|\mathcal{G}|^2 \times |\mathcal{G}|^2$ square on $u_n \in U$\\
    \For {$1 \leq i \leq 4$}{
        \For {$g_k \in G_i$} {
            Print exterior glue tile for $g_k$ on $u_n$ \\
            Print fan-out zone of length $|G|$ associated with $g_k$ on $u_n$\\
            Print fan-in zone of length $|G|$ associated with $g_k$ on $u_n$\\
            \tcc{Print signals along the frame associated with $g_k$ to glues which $g_k$ triggers}
            \For{$g_l \in \delta_n(g_k, i)$ where $g_l \in G_i^{\prime}$}{
                PrintPathToFrame($g_k$,$g_l$)\\
                PrintPathAroundFrame($g_k$, $g_l$)\\
                PrintPathFromFrame($g_k$, $g_l$)\\
                }
            }
        }
    }
\end{algorithm}
\caption{The main algorithm making macrotiles.}
\label{fig:mainMacrotileCreation}
\end{figure}

\subsubsection{Achieving equivalent dynamics}

We call the tiles that form a single tile wide frame around a macrotile the {\em boundary}. Tiles of a macrotile that are not on the boundary are called {\em interior} tiles. Glues on tile edges facing the outside of a boundary
are called {\em exterior} glues and glues on edges facing the inside of the boundary are called {\em interior} glues. To simulate both the production {\em and} the dynamics of an STAM$^+$ system, it is crucial that the boundary tiles of the simple macrotiles resulting from our fan-out reduction procedure form completely before playing a part in the assembly of the simulated system. Otherwise, attachment to an existing
assembly of partially formed macrotiles could lead to a situation depicted in Figure~\ref{fig:improperGrowth} where tiles of one macrotile ($T_1$ and $T_2$ in Figure~\ref{fig:improperGrowth}) begin to grow from the West, while
tiles of another macrotile ($T_3$ and $T_4$ in Figure~\ref{fig:improperGrowth}) grow from the South. Also note that only the boundary needs to be formed, the tiles that grow the interior of the macrotiles are only used to pass signals.
Once interior tiles fall into place asynchronous signals can be passed. Since signal passing in the STAM$^+$ is asynchronous the fact that interior tiles may attach at anytime does not prevent this system
from exhibiting equivalent dynamics. See \ifabstract\else Section~\ref{sec:proofSketch} \fi for details on correctness.
\begin{figure}[htp]
\begin{center}
\includegraphics[width=2.5in]{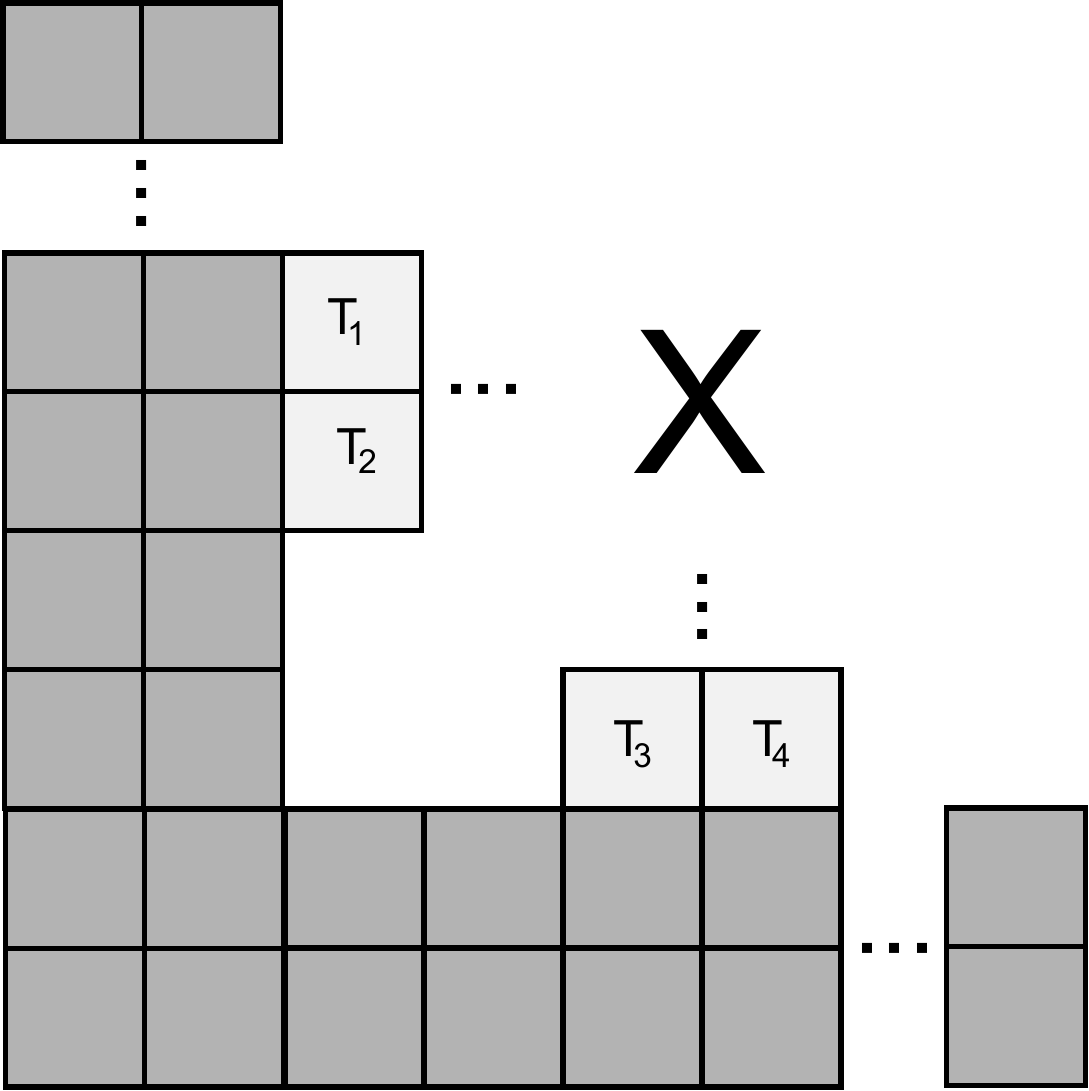}
\caption{Allowing exterior glues of macrotiles to bind prior to the formation of the boundary of the macrotile can prevent equivalent dynamics. For example, $T_1$ and $T_2$ may
be tiles of a macrotile that does not contain the tiles $T_3$ and $T_4$.}
\label{fig:improperGrowth}
\end{center}
\end{figure}

We provide a technique using signal passing that ensures that the boundary of a macrotile is formed prior to turning on glues exposed on the exterior of
the macrotile. This forces the macrotile's boundary to form before it is used in simulation at the expense of increasing the number of signals on a boundary tile by $3$. In Section~\ref{sec:reduceSigPerTile}
we show that this added signal complexity can be reduced.
For a given macrotile, we define a {\em circuit} to be a path of active glues through this macrotile such that the path starts and ends at a preset initial tile and the signals along this path propagate
through adjacent tiles until the signal reaches the initial tile. See Figure~\ref{fig:circuits} for an example.
\begin{figure}[htp]
\begin{center}
\includegraphics[width=4.5in]{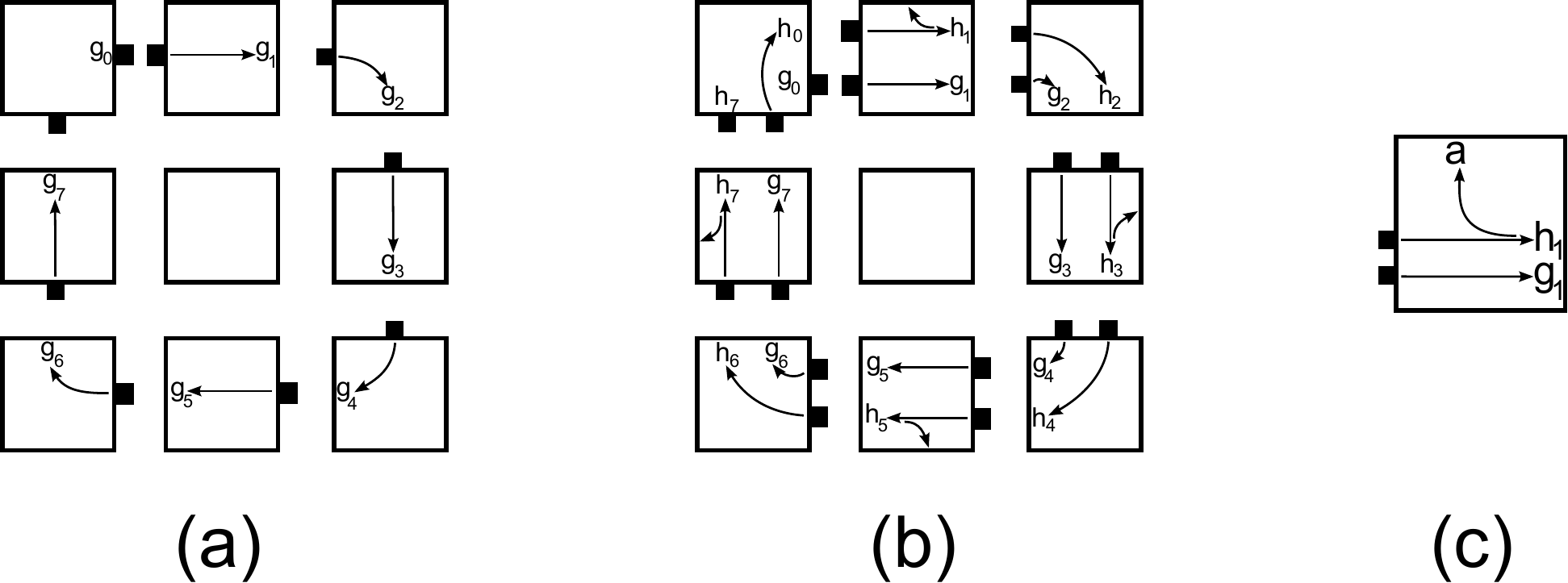}
\caption{{\bf(a)} An example of a circuit following the boundary tiles of a macro tile. {\bf(b)} An example of a verifying circuit labeled by g glues followed by a presenter circuit followed by h glues. Note that glue $g_7$ cannot initiate the presenter circuit until the boundary of the macrotile is fully formed. {\bf(c)} Glue `a' is turned on when $h_1$ binds. }
\label{fig:circuits}
\end{center}
\end{figure}
Now, overlay each of the border tiles with $2$ circuits both with initial tiles in the top left corner that follow a path going clockwise as
in Figure~\ref{fig:circuits}(b). One of these circuits, the {\em verifying circuit}, (labeled with g glues in Figure~\ref{fig:circuits}(b)) is used to check that the boundary tiles have completely formed.
The other circuit,
the {\em presenter circuit}, (labeled with h glues in Figure~\ref{fig:circuits}(b)) is used to turn on exterior glues of the macrotile. When glue $g_7$ of the verifying circuit has attached, it sends a signal to initiate the presenter circuit.
As signals are propagated through this circuit, we can once again utilize mutual activation to turn exterior glues on.
Figure~\ref{fig:circuits}(c) shows an example of turning on a North glue by
$h_1$ by mutual activation. As the presenter circuit propagates its signal, glue $h_1$ binds, sending a signal to turn \ton glue $a$. Each glue $h_i$ for $i < 8$
can use mutual activation to turn \ton an exterior glue of a macrotile. Since each tile belonging to the border of our constructed macrotile contains at most $1$ exterior glue, the number of signals
on any one of these boundary tiles increases by at most $3$. Finally, it is easy to see that a macrotile of any size can be endowed with these circuits and that
any macrotile with these circuits must have a completely formed boundary before exterior glues are turned \ton. We say that circuits with this property {\em delay exterior glues}.

\subsubsection{Circuit printing algorithm}
The circuit described above can be ``printed'' on a macrotile with the algorithm in Figure~\ref{fig:circuitAlgo}.

\begin{figure}[ht]
\begin{algorithm}[H]
  \KwData{A set of macrotile boundary tile types $B=\{t_i|1\leq i\leq m\}$ enumerated clockwise starting in the top left corner, $n=$scale factor}
  \KwResult{A set of macrotile boundary tile types equipped with circuits that delay exterior glues}
  \For{$0 \leq i < 4n-4$}{
    \tcc{handle initial tile type}
    \If{$i=1$}{
      \tcc{add glue $g_0$ to the East edge of type $t_0$ that is in the on state}
      addActiveGlue($t_i$, $g_0$, on, East)\\
      addActiveGlue($t_i$, $h_0$, latent, East)\\
      addActiveGlue($t_i$, $g_{m-1}$, on, South)\\
      addActiveGlue($t_i$, $h_{m-1}$, on, South)\\
      \tcc{add a signal from the South glue $g_{m-1}$ to the East glue $h_0$}
      addSignal($g_{m-1}$, South, $h_0$, East)
    }
    \tcc{North boundary tile type cases}
    \ElseIf{$1<i<n-1$}{
      \tcc{add glues and signal for the verifying circuit}
      addActiveGlue($t_i$, $g_{i-1}$, on, West)\\
      addActiveGlue($t_i$, $g_i$, latent, East)\\
      addSignal($g_{i-1}$, West, $g_i$, East)\\
      \tcc{add glues and signal for the presenter circuit}
      addActiveGlue($t_i$, $h_{i-1}$, on, West)\\
      addActiveGlue($t_i$, $h_i$, latent, East)\\
      addSignal($h_{i-1}$, West, $h_i$, East)

      \If{$\exists a$ an exterior glue of $t_i$}{
        \tcc{add a signal from the East glue $h_i$ to the exterior glue $a$}
        addSignal($h_{i}$, East, $a$, North)
      }
    }
    \ElseIf{$i=n$}{
      \tcc{add glues and signal for the verifying circuit}
      addActiveGlue($t_i$, $g_{i-1}$, on, West)\\
      addActiveGlue($t_i$, $g_i$, latent, South)\\
      addSignal($g_{i-1}$, West, $g_i$, South)\\
      \tcc{add glues and signal for the presenter circuit}
      addActiveGlue($t_i$, $h_{i-1}$, on, West)\\
      addActiveGlue($t_i$, $h_i$, latent, South)\\
      addSignal($h_{i-1}$, West, $h_i$, South)
    }
    \tcc{$\ldots$ continue for East, South and West boundary tile types and the two remaining corner tiles}
  }
\end{algorithm}
\caption{Algorithm to place a circuit on the boundary of a macrotile.}
\label{fig:circuitAlgo}
\end{figure}

\subsection{Reducing the number of signals per tile}\label{sec:reduceSigPerTile}

Any macrotile can be equipped with circuits that delay exterior glues. We would like to minimize the number of signals per tile. At the moment any tile in our macrotile construction contains at most $4$ signals. Here we present techniques to reduce the number of signals per tile required. First, we consider systems with any temperature and reduce the number of signals per tile to $2$. Then we show that for systems with temperature greater than $1$, cooperation allows us to
reduce the number of signals per tile required to just $1$.

\subsubsection{Reducing signals per tile at any temperature}
We can obtain the results of Theorem~\ref{thm:2SimpAtTemp1} by modifying
the macrotile construction given in Section~\ref{sec:MacroCreation}.
First, consider any macrotile given by our construction. Observe that tiles not on the boundary of a macrotile have at most $2$ signals already. However, tiles on the boundary may contain up to $4$ signals.
Figure~\ref{fig:4SimpSignalTile} shows a typical boundary tile. The two horizontal running signals belong to the verifying circuit and the presenter circuit. $d$ is the glue that is presented. In Figure~\ref{fig:4SimpSignalTile}, $d$ also signals some other glue. Therefore, there is a signal coming out of $d$ leading to the signal passing lanes of our macrotile construction.

\begin{figure}[htp]
\begin{center}
\includegraphics[width=1in]{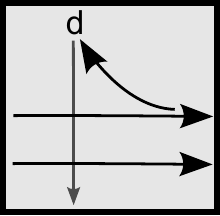}
\caption{Example of a boundary tile with $4$ signals.}
\label{fig:4SimpSignalTile}
\end{center}
\end{figure}

Note that the number of signals per tile can easily be reduced to $3$ by splitting up the verifying circuit and the presenter circuit. In this case,
we encompass a macrotile with a $2$ tile wide frame and
 the verifying circuit fires around the tiles just inside the boundary tiles. When the verifying circuit completes, it signals the presenter circuit to then fire around the boundary of the macrotile. This could lead to a $1$ tile wide gap between some macrotiles in an assembly, but this is enough to ensure proper growth since no other macrotile could fit into or start growing in this one tile wide gap.
 After placing verifying circuits and presenter circuits on separate tiles,
 reducing the number of signals per tile needed can be done using the gadget in Figure~\ref{fig:temp1BoundarySignalReduce}.

\begin{figure}[htp]
\begin{center}
\includegraphics[width=2.5in]{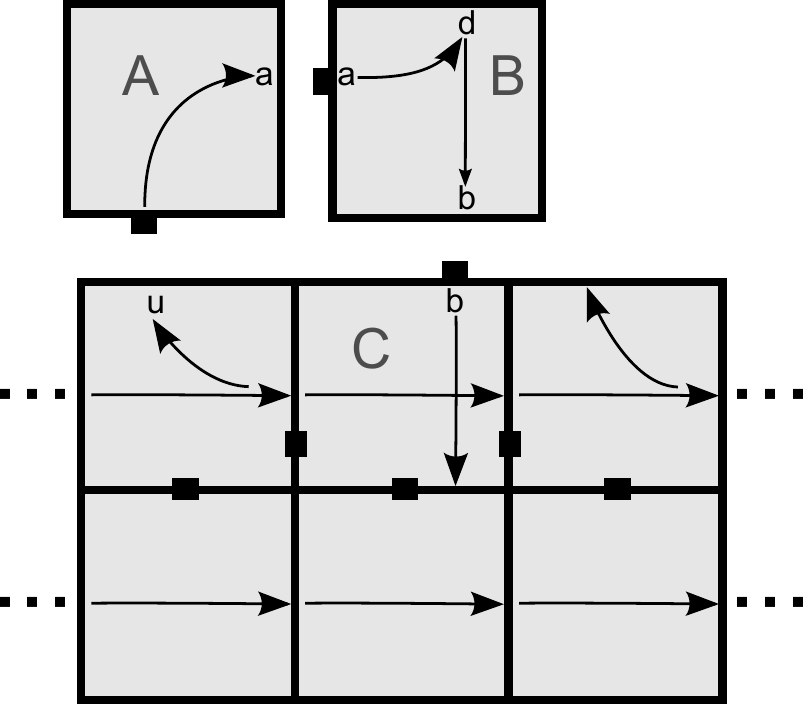}
\caption{Reducing the number of signals per tile to $2$.}
\label{fig:temp1BoundarySignalReduce}
\end{center}
\end{figure}

In Figure~\ref{fig:temp1BoundarySignalReduce}, the presenter circuit now turns on $u$, allowing tile $A$ to bind and turning on glue $a$. This allows tile $B$ to bind and eventually signal $d$ on. Now when $d$ sends a signal down, it only crosses one circuit per tile. This reduces the $4$ signals used in Figure~\ref{fig:4SimpSignalTile} to just $2$ signals per tile. In order to use this gadget with our macrotile construction, we must double the scale factor since we must add one tile between each case of fan-out.

Lemma~\ref{lem:oneToTwo} implies that at temperature $1$, $2$-simplified is the best that we can do. At temperature greater than $1$ cooperation can be used to reduce the number of signals per tile to $1$.

\subsubsection{Reducing signals per tile at temperature greater than $1$}\label{sec:reduceSigPerTileTemp2}
For simplicity, we present the details at temperature $2$. We modify the macrotile construction given in Section~\ref{sec:MacroCreation} to obtain the results of Theorem~\ref{thm:1SimpAtTemp2}.
By increasing the scale factor by the number of strength $2$ glues, we can assume that every exposed glue of a completely formed macrotile is of strength $1$ by representing strength $2$ glues with two strength $1$ glues. In the following modification to
our macrotile construction, we assume this to be the case.
The goal of this modification to the construction of our macrotiles is to reduce the number of signals per tile needed to only $1$.

First, notice that at temperature $2$, we can do away with the verifying circuit and the presenter circuit. This is done using cooperation. Consider a macrotile given by our construction prior to adding the these circuits, we can give all but one exposed exterior edge strength $1$ glues and the remaining exposed exterior edge a strength $2$ glue. Then, we can verify that a macrotile has a formed boundary with tiles that traverse the macrotile using cooperation. We give the final tile of these ``verifying'' tiles a glue of strength $2$ so that when it binds, it can allow for ``presenter'' tiles to grow around the macrotile using cooperation. Since all temperature $2$ glues on the exposed edges of the boundary of a completed macrotile have been represented by temperature $1$ glues, these glues cannot bind until the presenter tiles start to form a boundary. This prevents cases of improper growth depicted in Figure~\ref{fig:improperGrowth}.

With the verifying circuit and the presenter circuit gone, notice that macrotiles have at most $2$ signals per tile. There are three situations where $2$ signals are used: fan-in tiles, fan-out tiles and tiles with crossing signals. To achieve an STAM$^+$ system that is $1$-simplified it suffices to show that we can replace any of these three with constructions that only use $1$ signal per tile.
Fan-in can be replaced by the gadget in Figure~\ref{fig:temp2FanInSignalReduce}.
\begin{figure}[htp]
\begin{center}
\includegraphics[width=2.5in]{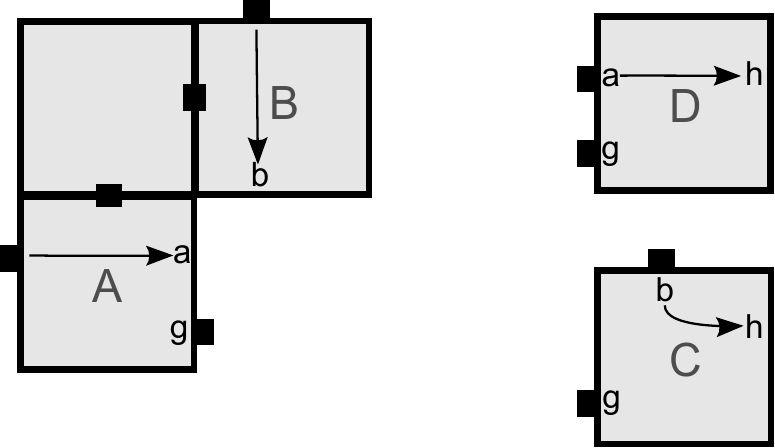}
\caption{Reducing the number of signals per tile to $1$ for fan-in at temperature $2$. $a$, $b$ and $g$ are strength $1$ glues. If glue $a$ of $A$ is signaled to turn \ton, cooperation allows tile $D$ to bind and a signal is sent to glue $h$ of $D$. Similarly, if glue $b$ of $C$ is turned \ton, tile $C$ can bind using cooperation and a signal is sent to glue $h$ of $C$. }
\label{fig:temp2FanInSignalReduce}
\end{center}
\end{figure}

\noindent Fan-out can be replaced by the gadget in Figure~\ref{fig:temp2FanOutSignalReduce}.

\begin{figure}[htp]
\begin{center}
\includegraphics[width=2in]{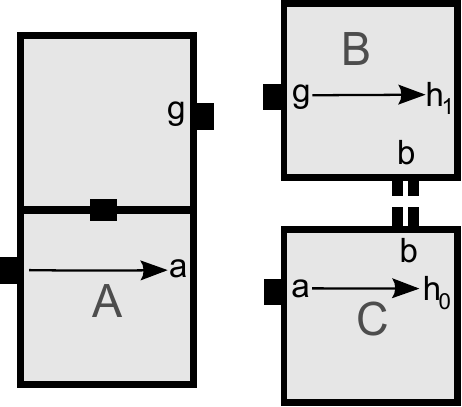}
\caption{Reducing the number of signals per tile to $1$ for fan-in at temperature $2$. $b$ is strength $2$ glue. $a$ and $g$ are strength $1$ glues. Therefore, tiles $B$ and $C$ can bond to form tuples. If glue $a$ of tile $A$ is turned \ton, one of these tuples can attach along glues $a$ and $g$. The binding of the $a$ glues signals $h_0$, and the binding of the $g$ glues signals $h_1$.}
\label{fig:temp2FanOutSignalReduce}
\end{center}
\end{figure}

\noindent All of the signal crossing that occurs in our macrotiles is of the form depicted by the single tile on the left in Figure~\ref{fig:temp2CrossingSignalReduce}. In particular, one signal runs vertically and the other runs horizontally. Signal crossings of this form can be replace by the gadget in Figure~\ref{fig:temp2CrossingSignalReduce}.

\begin{figure}[htp]
\begin{center}
\includegraphics[width=4in]{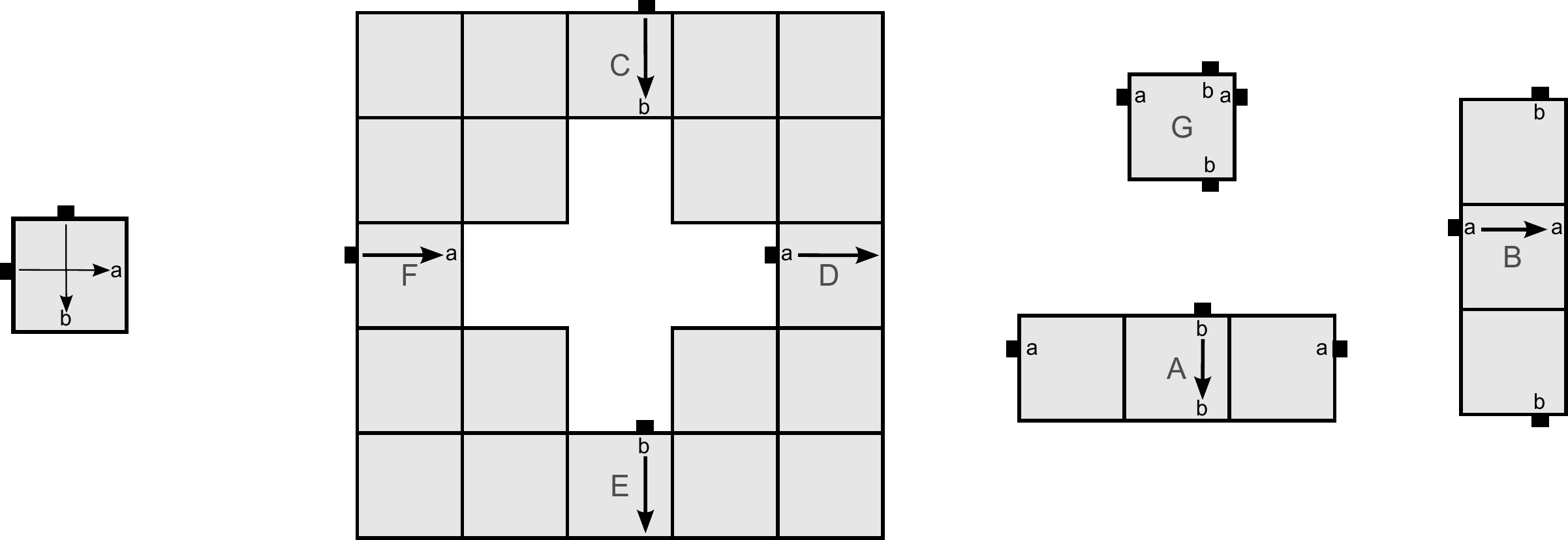}
\caption{Reducing the number of signals per tile to $1$ for signal crossing at temperature $2$. Left: a tile with signal crossing. Middle: a replacement gadget for signal crossing at temperature $2$. Right: two tiles that allow for signal crossing. $a$ and $b$ are strength $1$ glues. If $a$ is signaled to turn on, cooperation allows for tile $A$ to attach. The binding event along the East glue of $A$ allows a signal to continue horizontally. The signal on $A$ and the tile $G$
allow the vertical signal to continue. Likewise, if $b$ is signaled to turn on then the $B$ tile attaches, propagating the signal vertically, while the signal of $B$ and the tile $G$ allow for the horizontal signal to continue.}
\label{fig:temp2CrossingSignalReduce}
\end{center}
\end{figure}

Notice that to use these temperature $2$ gadgets, we first increase the scale factor so that we can represent a temperature $2$ glue by temperature $1$ glues. Also, when we replace tiles by these gadgets, the scale factor
must increase by a constant factor to preserve some connectivity of the macrotile. For example, as depicted in Figure~\ref{fig:temp2CrossingSignalReduce}, a tile with signal crossing is replaced by a $5$ by $5$ macrotile.
It is also interesting to note that in the original construction complete assembly of the macrotile is possible prior to the sending of any signals. With the modified constructions, some of the assembly of the
macrotile cannot occur until after signals have been fired.

Finally, for any temperature $\tau>1$, we can achieve the same reductions of signal complexity as follows. First, for any glues not on the exterior of the macrotile, replace strength $1$ glues in
the modification by strength $\lceil{\tau/2}\rceil$ glues and strength $2$ glues by strength $\tau$ glues. For glues on the exterior of the macrotile, just as we represent strength $2$ glues by two strength $1$ glues on two separate tiles,
we replace any strength $\tau$ glues by strength $\lceil{\tau/2}\rceil$ and strength $\lfloor{\tau/2}\rfloor$ glues. We keep the strengths of the other glues on the exterior of the macrotile fixed. This increases the scale factor by the number of $\tau$ strength glues, but ensures that strength $\tau$ bonds do not occur independent of the attaching of a macrotile.

\subsection{Proof of correctness}\label{sec:ProofOfCorrectness}
\label{sec:proofSketch}
To show correctness of the fan-out reduction algorithm we must show that the resulting STAM$^+$ system simulates the original STAM$^+$ system by defining a representation function $R$ that maps
macrotiles of the simulating system $\mathcal{S}$ to tiles of the original system $\mathcal{T}$. For a macrotile $\alpha^{\prime} \in \prodasm{S}$,
tile $\alpha$ in $T$ and empty space $\varepsilon$, define $R$ as follows.
$R(\alpha^{\prime}) = \alpha$ if a single exterior glue of $\alpha^{\prime}$ is on and all latent exterior glues of $\alpha^{\prime}$ match glues of $\alpha$. Define $R$ to be $\varepsilon$ otherwise.
Then for our construction and choice of $R$,
when the boundary of $\alpha^{\prime}$ has formed and all latent exterior glues of $\alpha^{\prime}$ match glues of $\alpha$, $\alpha^{\prime}$ uniquely represents
$\alpha$. Let $R^*$ be the assembly representation function. To show that $\mathcal{S}$ simulates $\mathcal{T}$ we must show that $\mathcal{T}$ follows $\mathcal{S}$ and that $\mathcal{S}$
strongly models $\mathcal{T}$.

We first show that $\mathcal{T}$ follows $\mathcal{S}$.
Let $\gamma, \beta \in \prodasm{T}$ and $\gamma^{\prime}, \beta^{\prime}\in \prodasm{S}$ be such that $R^*(\gamma^{\prime}) = \gamma$ and $R^*(\beta^{\prime}) = \beta$.
If $\gamma^{\prime}$ binds to $\beta^{\prime}$ to give $\sigma^{\prime}$ then the exposed edges of $\gamma^{\prime}$ and $\beta^{\prime}$ have glues that must match a subset of exposed glues of the assemblies
$\gamma$ and $\beta$. Therefore $\gamma$ and $\beta$ can bind along edge pairs that correspond to the edge pairs of $\gamma^{\prime}$ and $\beta^{\prime}$
whose binding gives $\sigma^{\prime}$. Letting $\sigma$ be this binding of $\gamma$ and $\beta$, we see that $R^*(\sigma^{\prime}) = \sigma$. Therefore, $\mathcal{T}$ follows $\mathcal{S}$.

To show that $\mathcal{S}$ models $\mathcal{T}$, let $\alpha$ and $\beta$ be in $\prodasm{T}$ and let $\gamma$ be in $C_{\alpha,\beta}^\tau$.
Let $\alpha^{\prime}, \beta^{\prime}  \in \prodasm{S}$ such that $R^*(\alpha^{\prime}) = \alpha$ and $R^*(\beta^{\prime}) = \beta$.
Note that in order for $R^*$ to map $\alpha^{\prime}$ each macrotile used in the assembly of $\alpha^{\prime}$ must have a completely formed boundary.
Likewise for $\beta^{\prime}$.
Then notice that exterior glues of $\alpha$ that binds with an
exterior glues of $\beta$ to yield $\gamma$ correspond to exterior glues of $\alpha^{\prime}$ and $\beta^{\prime}$. Once the glues of $\alpha^{\prime}$ and $\beta^{\prime}$ have been triggered,
allowing these glues to bind yields an assembly
$\gamma^{\prime} \in \prodasm{S}$ such that $R^*(\gamma^{\prime}) = \gamma$. Hence $\mathcal{S}$ models $\mathcal{T}$.

} %
\subsection{Summary of Results}

At temperature $1$, the minimum signal complexity obtainable in general is $2$ and while it is possible to eliminate either fan-in or mutual activation,
it is impossible to eliminate both. For temperatures greater than $1$, cooperation allows for signal complexity to be reduced to just $1$ and for
both fan-in and mutual activation to be completely eliminated. Table~\ref{tbl:summary} gives a summary of these two cases of reducing signal complexity and shows the cost of such reductions in terms of
scale factor and tile complexity.
\begin{table}
\centering
\setlength{\tabcolsep}{.5em}
\begin{tabular}{|c|c|c|c|c|}\hline\rowcolor{black!20!white}
Temperature & Signal & Scale Factor & Tile Complexity & Contains Fan-Out / \\\rowcolor{black!20!white}
 & per Tile & & & Mutual Activation\\\hline
$1$ & $2$ & $O(|T|^2)$   & $O(|T|^2)$   & one or the other    \\\hline
$>1$ & $1$ & $O(|T|^2)$   & $O(|T|^2)$   & neither \\\hline
\end{tabular}
\caption{The cost of reducing signal complexity at $\tau=1$ and at $\tau > 1$.\label{tbl:summary}}
\end{table}

\section{A 3D 2HAM Tile Set which is IU for the STAM$^+$}\label{sec:3D-IU-for-STAM}
In this section we present our main result, namely a 3D 2HAM tile set which can be configured to simulate any temperature $1$ or $2$ STAM$^+$ system, at temperature $2$.  It is notable that although three dimensions are fundamentally required by the simulations, only two planes of the third dimension are used.

\begin{theorem}\label{thm:3D-IU-for-STAM}
There is a 3D tile set $U$ such that, in the 2HAM, $U$ is intrinsically universal at temperature $2$ for the class of all 2D STAM$^+$ systems where $\tau \in \{1,2\}$.  Further, $U$ uses no more than $2$ planes of the third dimension.
\end{theorem}

To prove Theorem~\ref{thm:3D-IU-for-STAM}, we let $\mathcal{T}' = (T',S',\tau)$ be an arbitrary STAM$^+$ system where $\tau \in \{1,2\}$.  For the first step of our simulation, we define $\mathcal{T} = (T,S,\tau)$ as a $2$-simplified STAM$^+$ system which simulates $\mathcal{T}'$ at scale factor $m' = O(|T'|^2)$, tile complexity $O(|T'|^2)$, as given by Theorem~\ref{thm:2SimpAtTemp1}, and let the representation function for that simulation be $R': B^T_{m'} \dashrightarrow T'$.  We now show how to use tiles from a single, universal tile set $U$ to form an initial configuration $S_{\mathcal{T}}$ so that the 3D 2HAM system $\mathcal{U}_{\mathcal{T}} = (U,S_{\mathcal{T}},2)$ simulates $\mathcal{T}$ at scale factor $m = O(|T|\log|T|)$ under representation function $R: B^U_m \dashrightarrow T$.  This results in $\mathcal{U}_{\mathcal{T}}$ simulating $\mathcal{T}'$ at a scale factor of $O(|T'|^4\log(|T'|^2))$ via the composition of $R$ and $R'$.  Note that throughout this section, $\tau$ refers to the temperature of the simulated systems $\mathcal{T}$ and $\mathcal{T}'$, while the temperature of $\mathcal{U}_{\mathcal{T}'}$ is always $2$.

\subsection{Construction overview}

\ifabstract
\else
For clarity, we present a high-level overview of the construction.  Please see Section~\ref{sec:3D-construction-details} for more details.
\fi

Assuming that $T$ is a $2$-simplified STAM$^+$ tile set derived from $T'$, we note that for each tile in $T$: 1. glue deactivation is not used, 2. it has $\le 2$ signals, 3. it has no fan-out, and 4. fan-in is limited to $2$.
To simulate $\mathcal{T}$, we create an input supertile $\sigma_{\mathcal{T}}$ from tiles in $U$ so that $\sigma_{\mathcal{T}}$ fully encodes $\mathcal{T}$ in a rectangular assembly where each row fully encodes the definition of a single tile type from $T$.  Beginning with an initial configuration containing an infinite count of that supertile and the individual tile types from $U$, assembly begins with the growth of a row on top of (i.e. in the $z=1$ plane) each copy of $\sigma_{\mathcal{T}}$.  The tiles forming this row nondeterministically select a tile type $t \in T$ for the growing supertile to simulate, allowing each supertile the possibility of simulating exactly one $t \in T$, and each such $t$ to be simulated.  Once enough tiles have attached, that supertile maps to the selected $t$ via the representation function $R$, and at this point we call it a \emph{macrotile}.

\begin{figure}[htp]
\begin{center}
    \includegraphics[width=.60\textwidth]{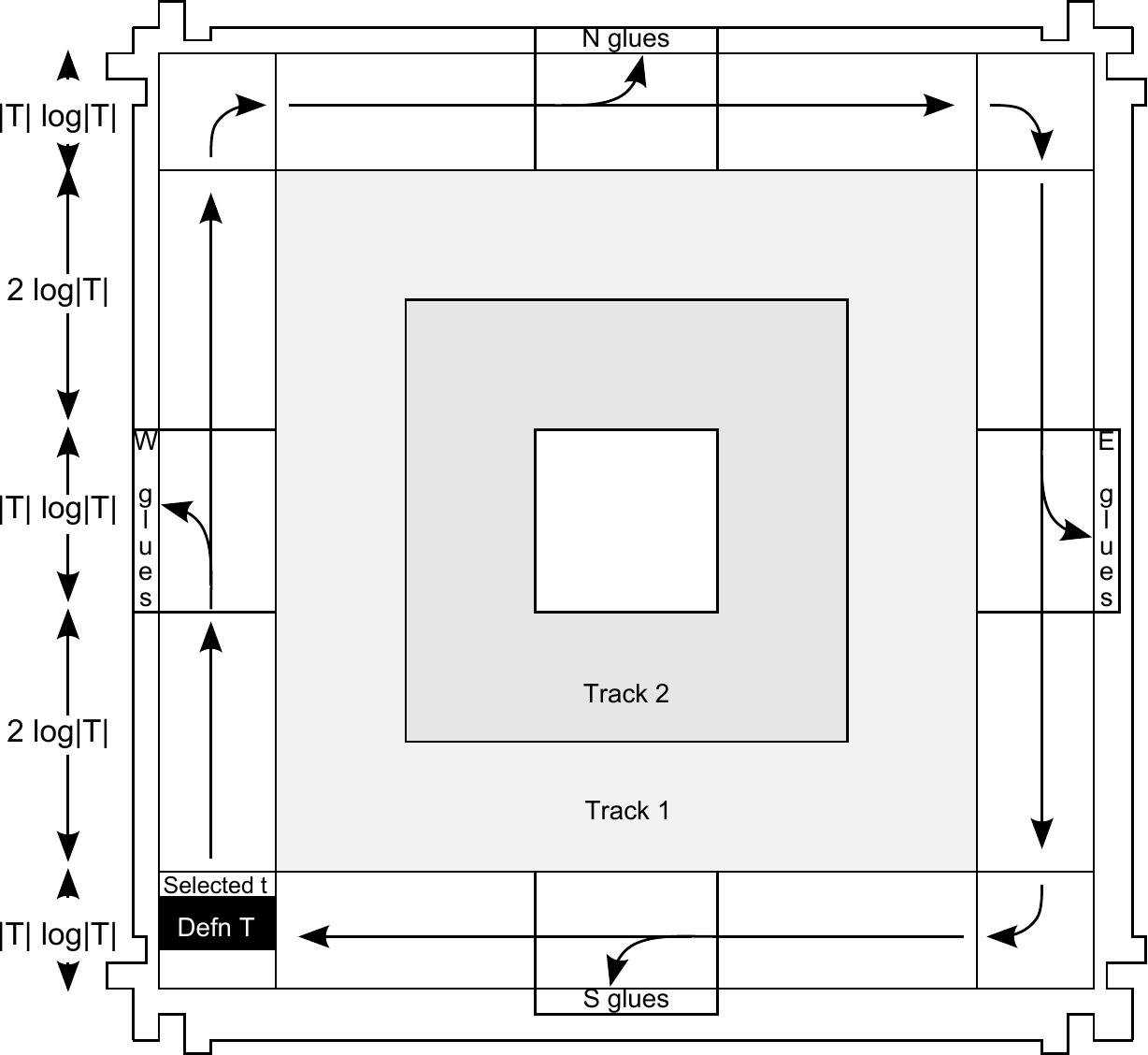}
    \caption{A high level sketch of the components and formation of a macrotile, including dimensions, not represented to scale.}
    \label{fig:3d-sim-macrotile-sketch}
\end{center}
\end{figure}

Each such macrotile grows as an extension of $\sigma_{\mathcal{T}}$ in $z=0$ to form a square ring with a hole in the center.  The growth occurs clockwise from $\sigma_{\mathcal{T}}$, creating the west, north, east, then south sides, in that order.  As each side grows, the information from the definition of $t$ which is relevant to that side is rotated so that it is presented on the exterior edge of the forming macrotile.  The second to last stage of growth for each side is the growth of geometric ``bumps and dents'' near the corners, which ensure that any two macrotiles which attempt to combine along adjacent edges must have their edges in perfect alignment for any binding to occur.  The final stage of growth for each side is to place the glues which face the exterior of the macrotile and are positioned correctly to represent the glues which begin in the \ton state for that side.

Once the first side of a macrotile completes (which is most likely to be the west side, but due to the nondeterministic ordering of tile additions it could potentially be any side), that macrotile can potentially bind to another macrotile, as long as the tiles that they represent would have been able to do so in $\mathcal{T}$.  Whenever macrotiles do bind to each other, the points at which any binding glues exist allow for the attachment of duples (supertiles consisting of exactly 2 tiles) on top of the two binding tiles (in $z=1$).  These duples initiate the growth of rows in $z=1$ which move inward on each macrotile to determine if there is information encoded which specifies a signal for that simulated glue to fire.  If not, that row terminates.  If so, it continues growth by reading the information about that signal (i.e. the destination side and glue), and then growth continues which carries that information inward to the hole in the center of the macrotile.  Once there, it grows clockwise in $z=0$ until arriving at the correct side and glue, where it proceeds to initiate the growth of a row in $z=1$ out to the edge of the macrotile in the position representing the correct glue.  Once it arrives, it initiates the addition of tiles which effectively change the state of the glue from \tlatent to \ton by exposing the necessary glue(s) to the exterior of the macrotile.

The width of the center hole is carefully specified to allow for the maximum necessary $2$ ``tracks'' along which fired signals can travel, and growth of the signal paths is carefully designed to occur in a zig-zag pattern such that there are well-defined ``points of competition'' which allow two signals which are possibly using the same track to avoid collisions, with the second signal to arrive growing over the first, rotating toward the next inward track, and then continuing along that track.  Further, the positioning of the areas representing the glues on each edge is such that there is always guaranteed to be enough room for the signals to perform the necessary rotations, inward, and outward growth.  If it is the case that both signals are attempting to activate the same glue on the same side, when the second signal arrives, the row growing from the innermost track toward the edge of the macrotile will simply run into the ``activating'' row from the first signal and halt, since there is no need for both to arrive and in the STAM such a situation simply entails that signal being discarded.  (Note that this construction can be modified to allow for any arbitrary full-tile signal complexity $n$ for a given tile set by simply increasing the number of tracks to $n$, and all growth will remain correct and restricted to $z \in \{0,1\}$.)

This construction allows for the faithful simulation of $\mathcal{T}$ by exploiting the fact that the activation of glues by fired signals is completely asynchronous in the STAM, as is the attachment of any pair of supertiles, and both processes are being represented through a series of supertile binding events which are similarly asynchronous in the 2HAM.  Further, since the signals of the STAM$^+$ only ever activate glues (i.e. change their states from \tlatent to \ton), the constantly ``forward'' direction of growth (until terminality) in both models ensures that the simulation by $\mathcal{U}_{\mathcal{T}}$ can eventually produce representations of all supertiles in $\mathcal{T}$, while never generating supertiles that don't correctly map to supertiles in $\mathcal{T}$ (equivalent production), and also that equivalent dynamics are preserved.

\begin{theorem}\label{thm:3D-IU-for-STAM-tau}
For each $\tau > 1$, there is a 3D tile set $\widehat{U}_{\tau}$ such that, in the 2HAM, $\widehat{U}_{\tau}$ is IU at temperature $\tau$ for the class of all 2D STAM$^+$ systems of temperature $\tau$.  Further, $U$ uses no more than $2$ planes of the third dimension.
\end{theorem}

To prove Theorem~\ref{thm:3D-IU-for-STAM-tau}, we create a new tile set $\widehat{U}_{\tau}$ for each $\tau$ from the tile set of Theorem~\ref{thm:3D-IU-for-STAM} by simply creating $O(\tau)$ new tile types which can encode the value of the strength of the glues of $T$ in $\sigma_{\mathcal{T}}$, and which can also be used to propagate that information to the edges of the macrotiles.  For the exterior glues of the macrotiles, just as strength $2$ glues were split across two tiles on the exterior of the macrotiles, so will $\tau$-strength glues, with one being of strength $\lceil \tau/2 \rceil$ and the other $\lfloor \tau/2 \rfloor$.  All glues which appear on the interior of the macrotile are changed so that, if they were strength $1$ glues they become strength $\lceil \tau/2 \rceil$, and if they were strength $2$ they become strength $\tau$.  In this way, the new tile set $\widehat{U}_{\tau}$ will form macrotiles exactly as before, while correctly encoding glues of strengths $1$ through $\tau$ on their exteriors, and the systems using it will correctly simulate STAM$^+$ systems at temperature $\tau$.

\later{
\section{Details of 3D 2HAM Simulation of STAM$^+$ Systems}\label{sec:3D-construction-details}

In this section, we present the details of the construction used to prove Theorem~\ref{thm:3D-IU-for-STAM} in the following steps:  (1) the construction of the $\sigma_{\mathcal{T}}$ , (2) the formation of macrotiles from $\sigma_{\mathcal{T}}$, (3) the passing of signals and activation of glues on a macrotile, and 4) a discussion of how the system $\mathcal{U}_{\mathcal{T}'}$ correctly simulates $\mathcal{T}$ and thus $\mathcal{T}'$.  Note that, except when specifically mentioned, all growth of the macrotile occurs in the plane $z=0$.

\subsection{The initial configuration}\label{sec:seed-supertile}

If the initial configuration of $\mathcal{T}$, $S$, is the default configuration (i.e. containing infinite copies of the singleton tiles from $T$ and nothing else), then the initial configuration for $\mathcal{U}_{\mathcal{T}}$, $S_{\mathcal{T}}$, will simply be infinite copies of each of the tiles from $U$ and the input supertile $\sigma_{\mathcal{T}}$.  If $S$ contains other supertiles, then $S_{\mathcal{T}}$ will also contain pre-formed supertiles which consist of macrotiles in $\mathcal{U}_{\mathcal{T}}$ which map to those supertiles using $R$.  For the rest of this discussion, we will assume the default configuration for $\mathcal{T}$, since such other configurations are a straightforward addition.

We create the initial configuration of $\mathcal{U}$, $S_{\mathcal{T}} = U \cup \{\sigma_{\mathcal{T}}\}$, where $\sigma_{\mathcal{T}}$ is a rectangle consisting of tiles from $U$ which has a single row encoding the definition of each tile $t \in T$, arranged so that their upward (those pointing in the positive $z$ direction) glues contain a representation of $t$, along with one additional row encoding the binary representation of the value $2\log |T|$ which will be used to control a binary counter later in the construction.  The content of the encoding can be seen in Figure~\ref{fig:3d-sim-tile-set-encoding}.

From left to right in each row (other than the northernmost two rows) is the encoding of the definition a tile type $t \in T$, prefaced with a `\#'. To encode tile definitions, first the set of all glues, $G$, used by tiles in $T$ are arranged in some order and each assigned a unique number from $1$ to $|G|$ (the number $0$ is reserved to represent the lack of a glue).  Then, for each side $d \in \{N,E,S,W\}$, a set of $|G|$ glue definitions is listed, separated by semi-colons.  This set consists of a glue definition of each glue as it appears (or doesn't appear) on side $d$ of $t$.  Each glue definition consists of 1) the strength of that glue on that side, followed by a comma, 2) a `$0$' if the initial state of that glue is \texttt{latent} and a `$1$' if its initial state is \texttt{on}, and 3) if that glue fires a signal, the letter corresponding to the destination side of that signal (i.e. $\{N,E,S,W\}$), then the binary number representing the number of glues to count past on the destination side before arriving at the glue activated by that signal (or $0$ for none).  Note that all binary number representations of glues are padded with leading $0$'s if necessary so that their lengths are all $\lfloor \log |G| \rfloor + 1$, including the $0$ representing no glue.  Additionally, rather than being prefaced by a `\#', the final tile type encoding (on the second most northern row) is prefaced with a `!'.  Finally, the ordering of the glue definitions is such that they increase from lowest to highest glue number from the west to east on the north and south sides, and from north to south on the east and west side.  This will ensure that the glues are represented in adjacent locations on adjacent edges of abutting macrotiles.

Since each row contains one entry of width $O(\log|G|) = O(\log|T|)$ for each of the $|G| = O(|T|)$ glues for each direction, the width of each row is $O(|T|\log|T|)$.  The height of $\sigma_{\mathcal{T}}$ is $O(|T|)$.

\begin{figure}[htp]
\begin{center}
\includegraphics[width=3.5in]{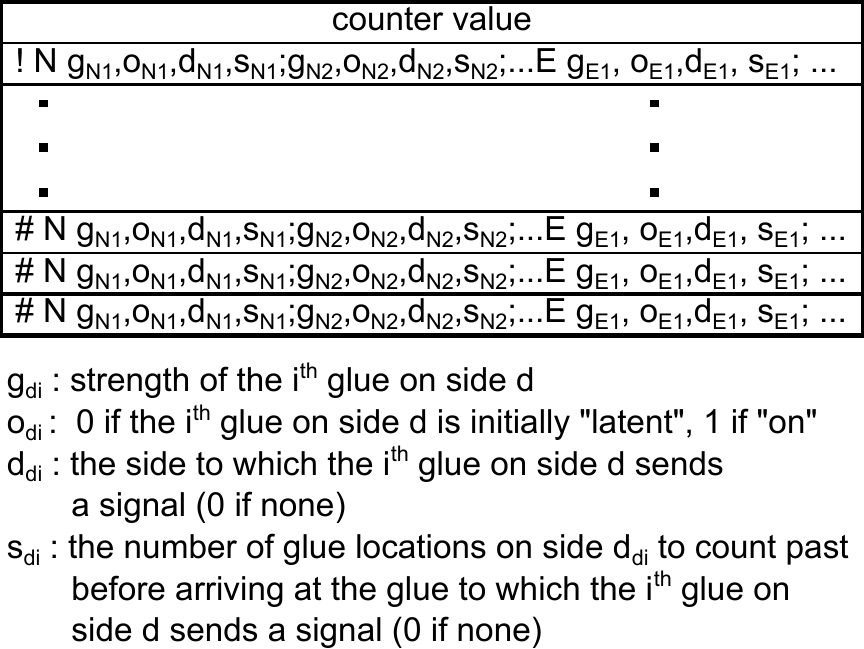}
\caption{The encoding of the simplified signal tile set for the 3D 2HAM simulation.}
\label{fig:3d-sim-tile-set-encoding}
\end{center}
\end{figure}

\subsection{Formation of macrotiles}

With the initial configuration of $\mathcal{S}$ consisting of infinite copies of the input supertile $\sigma_{\mathcal{T}}$ described in Section~\ref{sec:seed-supertile} and the singleton tile types of $U$, each copy of $\sigma_{\mathcal{T}}$ ``differentiates'' by first nondeterministically selecting a particular tile type $t \in T$ to represent.  This happens by a \emph{selection} row of tiles attaching to the top side (in the positive $z$ direction) of the southwest most tile, which encodes a `\#'.  This row grows north over the `\#' tiles, and initially the glues of the growing row encode the fact that no tile type has been selected.  As each `\#' is encountered, as long as no tile type has been selected, one of two tile types can attach.  One denotes that the tile is not chosen, and one that it is.  If the tile type is chosen, growth switches to a zig-zag pattern first moving east over the definition of the selected tile and copying that information up and to the north, then zagging back to the west and continuing until the northernmost row is reached.  Once a tile type has been selected, no additional tile type can be selected.  If no tile type has been selected by the time the `!' symbol is encountered (marking the beginning of the final tile type definition), then that tile type is selected.  In this way, it is guaranteed that exactly one tile type definition is marked as selected.  (Note that this selection process does not select each tile type with a uniform probability distribution, but is used for the simplicity of its explanation.  A process which achieves uniformity with arbitrary precision can be easily substituted.  See \cite{RNSSA} for description of such ``random number selection modules'' which could be substituted for this basic selection process.)  Finally, the row containing the definition of the selected tile grows over the counter value and also copies that information north one position, then, since it has now reached the northern edge of the supertile, down to the $z=0$ plane.

Following the return to $z=0$, growth continues to the north, following a zig-zag pattern of growth where each row grows completely from one side to the other, and then each subsequent row grows completely, starting from the side that the last finished and moving in the opposite direction.  This continues for $2\log|T|$ rows (guided by the counter value encoded by the seed supertile).  At this point, the values which encode the glues of west side of the selected tile type are rotated to the left, while the information for the rest of the tile continues to be propagated forward.  As soon as the rotation is complete, growth continues forward, again controlled by the counter value, then the entire row rotates to the east.  Growth continues clockwise around in a square, depositing the definitions of the glues for each side in the correct locations as shown in Figure~\ref{fig:3d-sim-macrotile-sketch}.  It is an important fact that, as the glue information for each side is rotated into position, the outermost three rows (those on the outermost edges of the macrotile, including those which encode the representations of the glues on the exterior of the macrotile) can't initially form.  When the path growing clockwise around the macrotile reaches the end of a side and begins to rotate to begin the formation of the next side, only then do the outermost $3$ rows of that side begin to form, starting from the corner and growing back along that side forming the bumps and dents near the corners (which can be seen in Figure~\ref{fig:3d-sim-macrotile-sketch}).  This growth happens in two passes, with the first pass creating the bumps and dents at both ends of the side, and then once that is complete the path turns back and forms the outermost row which can then potentially allow any glues on the exterior of the macrotile to be turned \texttt{on}.  The reason for ensuring that the bumps and dents form before any glue can be active on the exterior of the side of the macrotile is to guarantee that they serve as ``alignment tabs'', forcing any macrotiles which may bind across external glues to be correctly oriented with respect to each other, with the edges in complete alignment.  This prevents ``slippage'' where glues bind which are meant to represent different glues in $T$ due to misaligned macrotiles.

Note that the regions in the center of Figure~\ref{fig:3d-sim-macrotile-edge-layout} denoted as ``tracks'' remain empty during this stage of growth.  They will be utilized later in the construction.

\begin{figure}[htp]
\begin{center}
    \includegraphics[width=.90\textwidth]{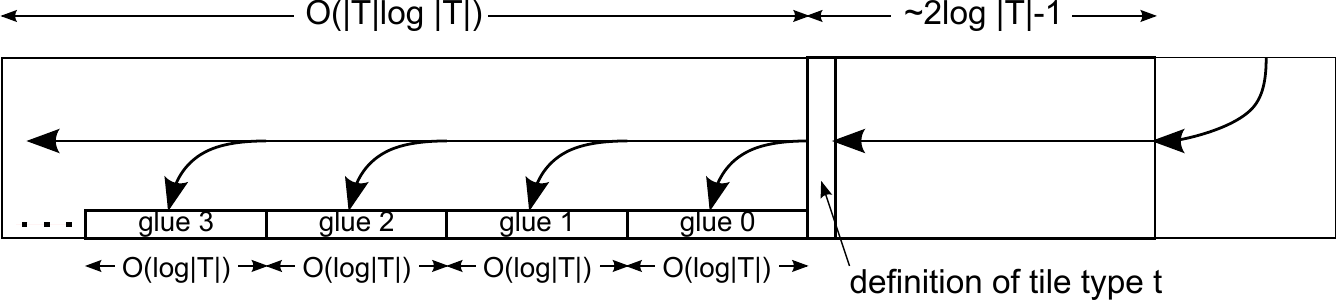}
    \caption{The formation of one side of a macrotile.}
    \label{fig:3d-sim-macrotile-edge-layout}
\end{center}
\end{figure}

The region of each side of a macrotile which is reserved for the glue definitions consists of a region for every glue in the full glue set $G$.  Each region is exactly wide enough to contain the encoding of the direction that a signal from that glue could be passed, the number of the glue that it could be passed to, and $3$ additional spaces on the right side.  See Figure~\ref{fig:3d-sim-macrotile-edge-layout} for a slightly more detailed depiction of the formation of a macrotile side, and Figure~\ref{fig:3d-sim-glue-states} for a depiction of how glues are represented in various states (to be discussed).

\subsection{Glue states and activation, and signal passing}

The encoding of each glue on the edge of a macrotile occupies the two outermost rows.  The innermost contains the information relevant to the signal (if any) which is initiated by the binding of the glue.  See Figure~\ref{fig:3d-sim-glue-states} for a visual depiction of how glues are represented. Of the three rightmost positions of the bottom row, the left and middle are reserved for tiles which will expose positive strength glues to the exterior of the macrotile if the glue in the glue position is simulating a glue in the \texttt{on} state.  If a \texttt{latent} glue is being simulated (or the lack of this glue on this side of the tile type being simulated by this macrotile), then none of the three rightmost positions of the bottom row will have tiles.  If $\tau=2$, only the center tile will expose a single strength-$1$ glue if a strength-$1$ glue is being simulated, and both the center and left tiles will expose strength-$1$ glues if a strength-$2$ glue is being simulated.  If $\tau=1$, both the center and the left tiles will expose strength-$1$ glues.  (In this way, although our simulator $\mathcal{U}_{\mathcal{T}}$ operates at temperature $2$, if $\tau=2$ we exactly mimic the strengths of the glues in $T$, but if $\tau=1$ we ensure that each external glue on a macrotile has sufficient strength to single handedly allow for the combination of macrotiles, since that is the case in $\mathcal{T}$).  In order for the tiles which represent the \texttt{on} glues to attach, a row grows in the $z=1$ plane over the bottom edge of the rightmost tile and provides the cooperation and information necessary to allow the attachment of either $1$ or $2$ tiles, as appropriate for the particular glue.  If such an \texttt{on} glue is on in the initial configuration of the tile type being simulated (i.e. without being turned \texttt{on} by a signal), then that row is initiated during the initial formation of that side during macrotile growth.  Otherwise, if that glue is turned \texttt{on} by a signal, then that row is initiated by growth which will be discussed shortly.

\begin{figure}[htp]
\begin{center}
    \includegraphics[width=.90\textwidth]{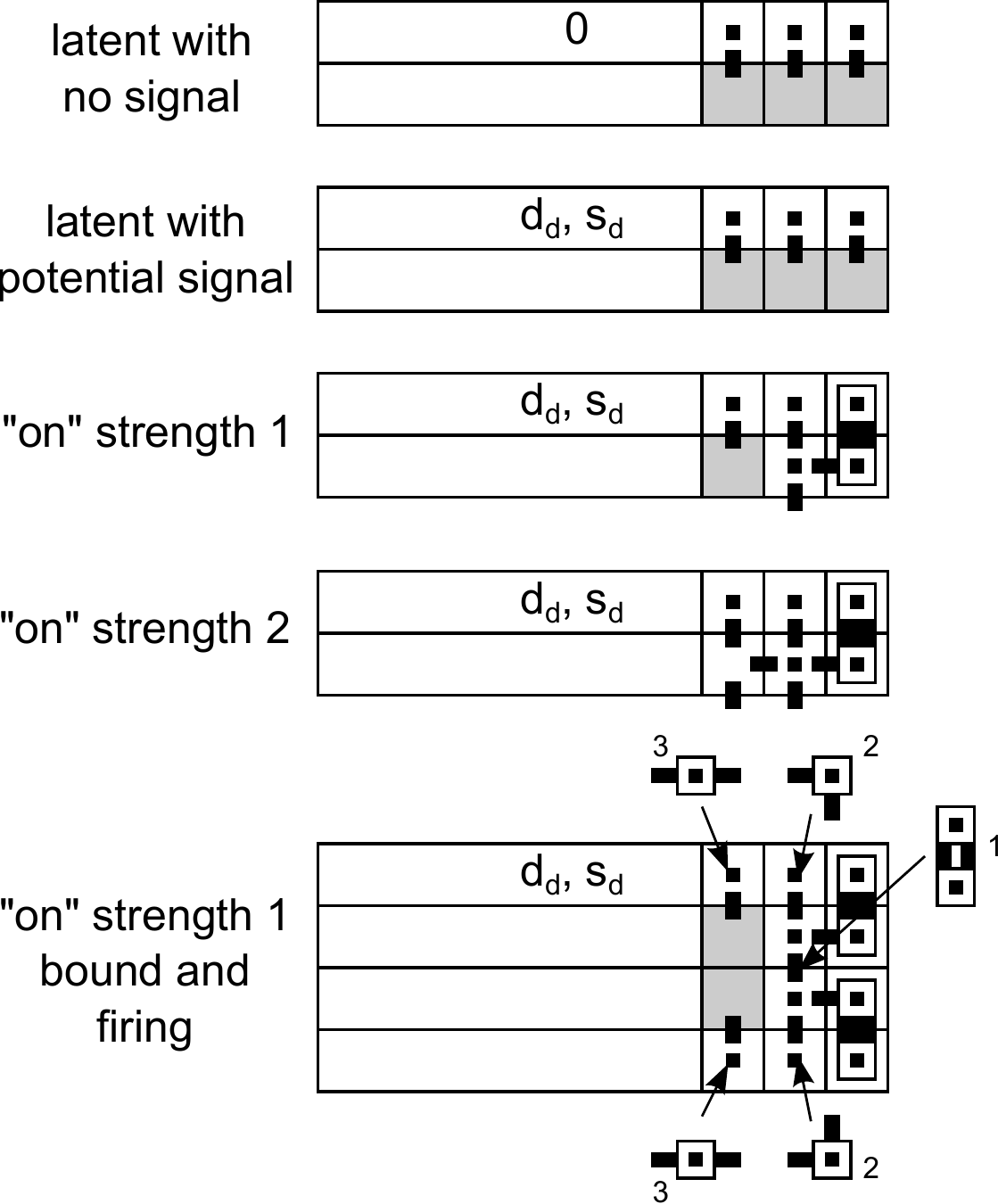}
    \caption{A series of examples showing how the states of external glues of macrotiles are represented.  Grey squares represent empty locations, and smaller white squares represent tiles in the $z=1$ plane while all others are in $z=0$.  Each black bar sticking out from a tile represents a strength $1$ glue, and a black square in the center of a tile represents a glue between the $z=0$ and $z=1$ planes.  Note that the glues are only shown for the three rightmost columns.}
    \label{fig:3d-sim-glue-states}
\end{center}
\end{figure}

If a glue in the \texttt{on} state is being represented, only then can it bind with the possibly adjacent, matching, and \texttt{on} glue of a neighboring macrotile.  If such a binding occurs, then (as shown in the bottommost picture of Figure~\ref{fig:3d-sim-glue-states}), and only then, can the duple (i.e. supertile composed of exactly two tiles) denoted by the number `$1$' attach across the adjacent macrotiles.  The attachment of this duple (in the $z=1$ plane) allows the growth of a row over to the region which denotes what signal (if any) should be fired by this glue binding event.  If a signal is in fact specified, a row of tiles in $z=1$ copy the information about that signal up and toward the center of the macrotile.  Otherwise, the growth of the row terminates.

The center of the macrotile consists of $2$ \emph{tracks}, which are simply empty ``lanes'' of width approximately $\log|T|$ each, around which paths of tiles representing signals can grow.  All signals grow in a clockwise direction around their tracks until reaching the destination side for their signals, and then find the location of the target glue by counting the number of glue positions passed until that number matches the target number they contain. (The information about glue locations is provided by the tiles they bind to around the track)  Upon reaching the location of the destination glue, they send a single-tile-wide path of tiles in $z=1$ across the portion of the macrotile between the signal and the outermost row of the macrotile, where they can then initiate the placement of the tiles that represent the glue in its \texttt{on} state. This path is the counterclockwise most column of the region associated with this glue, and is guaranteed to be either clear so that the glue can be activated, or already taken which means that some other signal has activated the glue (in the case of fan-in equal to $2$) and there is no purpose to this path continuing.  The final portion of this path can be seen in Figure~\ref{fig:3d-sim-glue-states} as the rightmost column.

Since there is no fan-out, each glue can fire at most one signal, and since any signal which activates it uses only the counterclockwise most column, the path in $z=1$ between the signal definition and the tracks is guaranteed to be free.  When a signal is fired, a zig-zag path in $z=1$ of width about $\log|T|$ (containing the signal's destination information) grows to the edge of ``track 1''.  An example can be seen as the grey path extending north from the $g_1$ region in the top image of Figure~\ref{fig:3d-sim-signals}.  Upon reaching the edge of track 1, it attempts to place a tile into the position shown as a green square in $z=0$, which is known as a ``point of competition''.  Due to the zig-zag growth of this path and also that of any signal which may already be occupying that track, if this path can place a tile there, it ``wins the competition'' for that track and then grows down into that track at $z=0$ and rotates to continue growth in the clockwise direction around the track, finding the correct side and glue destination for its signal, initiating the growth of the row which activates the row, and also continuing around the track until either making a full cycle or running into another signal which won a different point of competition.  Whenever it encounters another signal in front of it, since it has already delivered its signal it terminates growth.

In the case that it hasn't delivered its signal, when it runs into another signal occupying its track, it is able to briefly ``climb onto the back of'' that signal and rotate toward the center of the macrotile so that it can claim the inner track.  An example can be seen in the bottom image of Figure~\ref{fig:3d-sim-signals}.  It grows down into the inner track and continues clockwise.  Since the maximum number of signals possible in the macrotile is $2$, and there are $2$ tracks, the signal is guaranteed to have a free track to use.

In order to allow the construction to work solely in $2$ planes, namely $z=0$ and $z=1$, it is necessary to ensure that signals can never be prematurely blocked without climbing up beyond $z=1$.  This is guaranteed by the layout of the regions representing the glues on the edges of the macrotile, the position of the tracks, and the fact that all signals travel clockwise, and the coordinated zig-zag growth of the paths carrying signal information which allows for well-defined points of competition.  Examples of how the various situations are handled can be seen in Figure~\ref{fig:3d-sim-signals}.  The top image shows how glue $g_1$ could fire a signal which activates glue $g_2$ on the same side.  First, the path carrying the signal's information grows to track 1.  It then wins the point of competition shown by the green square, so can move down into the track. The side specified by the signal tells it that it is at the correct side, and the glue count tells it that it doesn't need to pass any signals but has already arrived at the correct location.  At that point, it sends the single-tile-wide path of signals down the rightmost column of glue $g_2$'s region and activates the glue. Additionally, the signal continues around the track until it runs into either itself or some other signal that has moved into the track.

In the middle image of Figure~\ref{fig:3d-sim-signals}, glue $g_2$ has now fired its own signal.  When it arrives at track 1, it has lost the point of competition for track 1 (red), so it grows over track 1 and acquires track 2.  At this point, it grows clockwise and is guaranteed to always have a signal in track 1 to grow over, since the signal for $g_1$ would either continue completely around the macrotile, or be blocked by another signal which would continue until it is blocked.

In the bottom image of Figure~\ref{fig:3d-sim-signals}, a signal which originated on another side of macrotile arrives in track 1, apparently having beat out the first signal for that track at some point.  This signal's destination is glue $g_3$, but it cannot make it that far in track 1 as it has lost the point of competition (red).  Therefore, it must briefly climb onto the back of the signal that beat it and move to track 2.  Because of the fact that there are $3$ rows on the right side of glue $g_1$'s region to the right of its signal definition, there is room for the losing signal to climb on top of the signal that it fired without blocking the path along the rightmost column.  This ensures that if a later signal wants to activate $g_1$, its path will not be blocked.  Further, since the width of the signal on the track and the width of the space used by the signal for $g_1$ are both the same, there is room for the signal to rotate northward without growing into the space reserved for $g_2$.  The signal that lost track 1 acquires track 2, then finally reaches the correct location to send its activation path to $g_3$.

Thus, the macrotile layout and pattern of growth ensure that no matter their originations or destinations, the maximum $2$ signals possible in any macrotile are guaranteed to always be able to use one of the $2$ tracks to successfully navigate to and activate their target glues, while never requiring growth into planes other than $z=0$ and $z=1$.  (Note that by simply increasing the number of tracks to some $n > 2$, $n$ signals could similarly be handled in a macrotile without collision or requiring more planes than $z \in \{0,1\}$.)

\begin{figure}[htp]
\begin{center}
\includegraphics[width=3.0in]{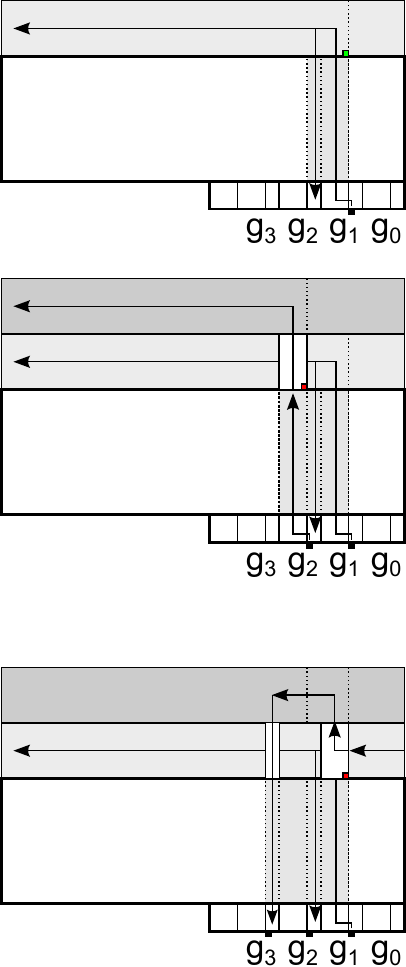}
\caption{A series of examples showing how signals are passed and received.}
\label{fig:3d-sim-signals}
\end{center}
\end{figure}

\subsection{Complexity analysis}

In order to simulate $\mathcal{T}$, $\mathcal{U}_{\mathcal{T}}$ requires building macrotiles of size $O(|T|\log|T|) \times O(|T|\log|T|) \times 2$.  Thus, the scale factor is $O(|T|\log|T|)$, while using a constant sized tile set and requiring only $2$ planes in 3D.  Since $\mathcal{T}$ simulates $\mathcal{T}'$ by using macrotiles of size $O(|T'|^2) \times O(|T'|^2)$ and also with tile complexity $|T| = O(|T'|^2)$, the scale factor for $\mathcal{U}_{\mathcal{T}}$ simulating $\mathcal{T}'$ is $O(|T'|^4 \log(|T'|^2))$.

The tile complexity is $O(1)$ as there is a single universal tile set $U$ which thus simulates any $T'$.

\subsubsection{Improved scale factor}
Note that by using the techniques of Theorem 8 in \cite{2HAMIU}, the scale factor for the simulation of $\mathcal{T}$ by $\mathcal{U}_{\mathcal{T}}$ can be reduced to $\sqrt{T} \log(|T|)$ for an overall scale factor of $O(|T'|^3 \log(|T'|^2))$ for the simulation of $\mathcal{T}'$. However, as that construction is much more complicated (for relatively small gain in scale factor), it is omitted from this version of the paper.

\subsection{Representation function and correctness of simulation} \label{sec:3D2HAMIUProof}

$R: B^U_m \dashrightarrow T$ is the representation function that maps supertiles in $\mathcal{U}_{\mathcal{T}}$ to tiles in $\mathcal{T}$. It does this by simply locating (in the nearly bottom-left corner of each macrotile) the subassembly corresponding to $\sigma_{\mathcal{T}}$, and then the row above it (in $z=1$) such that the leftmost tile type encodes that that row was selected.  This row contains the definition of the selected tile type which this macrotile does or will (if it hasn't yet grown any further) simulate.  If such a selection tile exists, $R$ simply reads the corresponding row and returns the tile type $t \in T$ encoded there.  If the tile doesn't exist, $R$ is undefined (i.e. this supertile maps to nothing).  We now discuss why the simulation of $\mathcal{T}$ is valid under this representation function.

Each supertile which has not yet grown to the point of making a tile selection, and thus does not map to any supertile in $\mathcal{T}$, is completely unable to interact with any other macrotiles as it has no external glues that can bind to another macrotile.  Further, each such supertile is guaranteed to eventually grow into a supertile which does map to a valid $t \in T$, and for every $t \in T$ a supertile which represents it can form (this is the base case for equivalent production).  Once a supertile has grown to the point that it maps to a $t \in T$, the ordering of growth, including the clockwise direction of growth by all macrotiles as well as the formation of each external edge in such a way that the geometric ``bumps'' and ``dents'' are formed before external glues are placed, ensures that the only interactions that can occur between macrotiles (or compositions of macrotiles) are at the well-defined locations representing the exteriors, and with the representations of the glues from $T$ in proper alignment.  Thus, even if a macrotile hasn't fully completed its initial growth to form a complete square, any interactions that it can have with other macrotiles must correctly follow the possible interactions between the tiles and supertiles from $T$ that they represent.  Further, due to the asynchronous nature of the STAM, and the fact that $\mathcal{T}$ is an STAM$^+$ system which does not include glue deactivation, any delays in the full formation of a macrotile and its external glues can only delay (but not permanently prevent) either the binding of macrotiles (which correctly models the asynchronous binding of supertiles in $\mathcal{T}$) or the firing of signals (which also correctly models the asynchronous firing of signals).  The design of the macrotiles and the tracks along which signals can travel ensure that (1) every time glues bind across macrotiles, any potential signals that may fire are in fact fired, and (2) all such signals will eventually arrive at their destinations and correctly activate their target glues.  Again, since signals are guaranteed to be correctly fired and in $\mathcal{T}$ their arrival time is indeterminate, the simulation by $\mathcal{U}_{\mathcal{T}}$ is correct.

Since every singleton tile from $T$ is correctly represented, along with their abilities to combine, fire activation signals, and activate glues, and $\mathcal{U}_{\mathcal{T}}$ cannot produce supertiles which either don't (and never will) represent supertiles in $\mathcal{T}$, or which fail to ultimately behave in the same ways and produce terminal assemblies in $\mathcal{U}_{\mathcal{T}}$ which don't map to terminal assemblies in $\mathcal{T}$, equivalent production holds and also $\mathcal{T}$ follows $\mathcal{U}_{\mathcal{T}}$.  Finally, since for every $\ta', \tb' \in \prodasm{U_{\mathcal{T}}}$ such that $\ta'$ and $\tb'$ are valid supertiles composed of macrotiles where $\tilde{R}(\ta') = \ta$ and $\tilde{R}(\tb') = \left(\tb\right)$ for $\ta, \tb \in \prodasm{T}$, if $\ta$ and $\tb$ can combine in $\mathcal{T}$, then $\ta'$ and $\tb'$ are guaranteed to be able to grow into (if necessary) $\ta''$ and $\tb''$ such that $\ta''$ and $\tb''$ can combine.  This may be the result of either the macrotiles finishing their initial formation, or the assembly of the paths representing the firing and propagation of signals.  This ensures that $\mathcal{U}_{\mathcal{T}}$ strongly models $\mathcal{T}$, and therefore $\mathcal{U}_{\mathcal{T}}$ simulates $\mathcal{T}$.

Finally, we note that by the simple composition of representation functions $R$ and $R'$, we can create a new representation function $R'': B^U_m \dashrightarrow T'$ such that $R''(\alpha) = R'(R^*(\alpha))$.  Thus, since $\mathcal{U}_{\mathcal{S}}$ simulates $\mathcal{T}$ under $R$ and $\mathcal{T}$ simulates $\mathcal{T}'$ under $R'$, $\mathcal{U}_{\mathcal{S}}$ simulates $\mathcal{T}'$ under $R''$.

}

\section{For Each $\tau$, an STAM$^+$ Tile Set Which is IU for the STAM$^+$ at $\tau$}

In this section, we show that for every temperature $\tau$ there is an STAM$^+$ tile set which can simulate any STAM$^+$ system at temperature $\tau$.  The constructions for $\tau > 1$ and $\tau = 1$ differ slightly, and are presented in that order, and both are based heavily upon the construction in Section~\ref{sec:3D-IU-for-STAM}.

Let $\mathcal{T'} = (T', S_\mathcal{T'}, \tau)$ be an arbitrary STAM$^+$ system and let
Let $\mathcal{T} = (T, S_\mathcal{T}, \tau)$ be the STAM$^+$ system constructed in Section~\ref{sec:transform-to-simple-stam} that is $1$-simplified in the case where $\tau > 1$ and $2$-simplified in the case where $\tau = 1$.
The constructions here yield STAM$^+$ systems $\mathcal{U}_{\mathcal{T}} = (U_{\tau}, S_{\mathcal{U}_{\mathcal{T}}}, \tau)$ that simulate each $\mathcal{T}$, and therefore prove Theorem~\ref{thm:STAM+IU}.

\begin{theorem}
\label{thm:STAM+IU}
For each temperature $\tau$, there is an STAM$^+$ tile set $U_{\tau}$ such that $U_{\tau}$ is intrinsically universal for the class of all STAM$^+$ systems at temperature $\tau$.
\end{theorem}

\subsection{STAM$^+$ Tile Sets which are IU for the STAM$^+$ at each $\tau > 1$}

This construction makes use of macrotiles very similar to those of the construction in Section~\ref{sec:3D-IU-for-STAM} to simulate the tiles of $\mathbb{T}$.  To see how they work, note that in the 2D STAM$^+$, growth of macrotiles can proceed in exactly the same way as most of the macrotile growth in the 3D 2HAM construction of Section~\ref{sec:3D-IU-for-STAM} with a few exceptions. First, since we cannot take advantage of multiple planes, the simulated signals must propagate using actual signals of active tiles. This can be done as follows.
First, we queue a signal to be fired along the necessary path through the macrotile. In the 3D 2HAM construction, as a macrotile grows, glues that can potentially propagate a simulated signal upon binding expose glues that allow for a duple to be placed above them in the $+z$-direction. In the STAM$^+$ construction, we can simply replace these tiles with signal tiles such that the binding of a glue would fire cascading signals toward a signal track of the macrotile. See figure~\ref{fig:zig-zag-macrotile} for an example of a signal track.
The entire sequence of firing a simulated signal, resulting in a glue going from $\texttt{latent}$ to $\texttt{on}$, is as follows.
Following the construction of Section~\ref{sec:3D-IU-for-STAM}, glues that allow for two macrotiles to bind are assumed to have strength less than $\tau$, and a strength $\tau$ glue $g$ is represented by two glues $g_1$ and $g_2$. This means that when we turn a glue $\texttt{on}$, we really expose glues on the exterior of a macrotile that allow for tiles to bind to the macrotile that \emph{present} either a single glue or two glues such as $g_1$ and $g_2$. (See figure~\ref{fig:3d-sim-glue-states} for the 3D 2HAM version of this presentation of glues.)
Let $\alpha$ be a macrotile producible in $\mathcal{U}_{\mathcal{T}}$, let $a$ be a glue of the simulated tile that fires a signal, and let $b$ (respectively $b_1$ and $b_2$ in the case that we are turning $\texttt{on}$ a strength $\tau$ glue) be a glue that turns $\texttt{on}$ as a result of the binding of $a$. When $\alpha$ binds to another macrotile $\beta$ (producible in $\mathcal{U}_{\mathcal{T}}$) using $a$, the binding event for $a$ initiates the propagation of signals that eventually expose a strength $\tau$ glue, $g_s$, on the inside of the signal track (which is initially empty). Edges of the macrotile that border the signal track expose glues that allow for cooperative binding in a clockwise manner around the track. All but one of these glues -- all but the glue on the tile abutting the signal track at the location directly across from the tile location for the simulated glue to be turned $\texttt{on}$ -- are the same. That final glue, call it $g_e$, initiates a cascade of signals that turns $\texttt{on}$ a glue $g_h$ on the exterior of $\alpha$.
Now, when a tile binds to $g_s$, a single tile wide path of tiles grows along the signal track using cooperation until binding with $g_s$. The binding event involving $g_s$ fires a signal that initiates the propagation of signals that expose $g_h$ on the exterior of $\alpha$. Then, $g_h$ allows for tiles to attach that present $b$ (respectively $b_1$ and $b_2$), completing the simulation of the signal from $a$ turning $\texttt{on}$ $b$. To finish the construction, we show how to obtain $S_{\mathcal{U}_{\mathcal{T}}}$ from $S_\mathcal{T}$.

\subsubsection{Creating the initial configuration}\label{sec:stam+seed_temp2}

Here we give a conversion of the initial configuration given in Section~\ref{sec:seed-supertile} that works for STAM$^+$ systems with $\tau > 1$. As in Section~\ref{sec:seed-supertile}, suppose that the initial configuration $S_\mathcal{T}$ is the default configuration. We create the initial configuration $S_{\mathcal{U}_{\mathcal{T}}}$ by letting $S_{\mathcal{U}_{\mathcal{T}}} = U_{\tau} \cup \{\sigma_\mathcal{T}\}$ where $\sigma_\mathcal{T}$ is a rectangle consisting of tiles from $U_{\tau}$ which has a single row encoding the definition of each tile $t\in T$. The specific encoding that is used is identical to the encoding used in Section~\ref{sec:seed-supertile}. Figure~\ref{fig:stam+iu_seed_gadget} depicts one of these rectangular supertiles, along with selection tiles. To select a row from one of these rectangular supertiles, tiles cooperatively bind to the east edge of tiles of the rectangle and nondeterministically choose a row. In figure~\ref{fig:stam+iu_seed_gadget}, the tiles labelled $N$ are skipping rows, the tile labelled $Y$ selects the row and the tiles labelled $D$ mark the fact that a row has been selected. Once again, care must be taken to ensure that some row is always selected. Once a row is selected, the binding of glues $g_s$ in figure~\ref{fig:stam+iu_seed_gadget} fire a signal that turns $\texttt{on}$ either a glue $g_0$ or a glue $g_1$ depending on the row and the encoding of the tile $t \in T$ and also turns $\texttt{on}$ another $g_s$ glue should the signal need to be propagated. The binding of $g_0$ or $g_1$ glues fires signals that propagate to the north edge of the tiles located on the top of the rectangular supertile, thus completing the nondeterministic selection of some tile $t \in T$ for the macrotile to simulate.

\begin{figure}[htp]
\begin{center}
\includegraphics[width=4in]{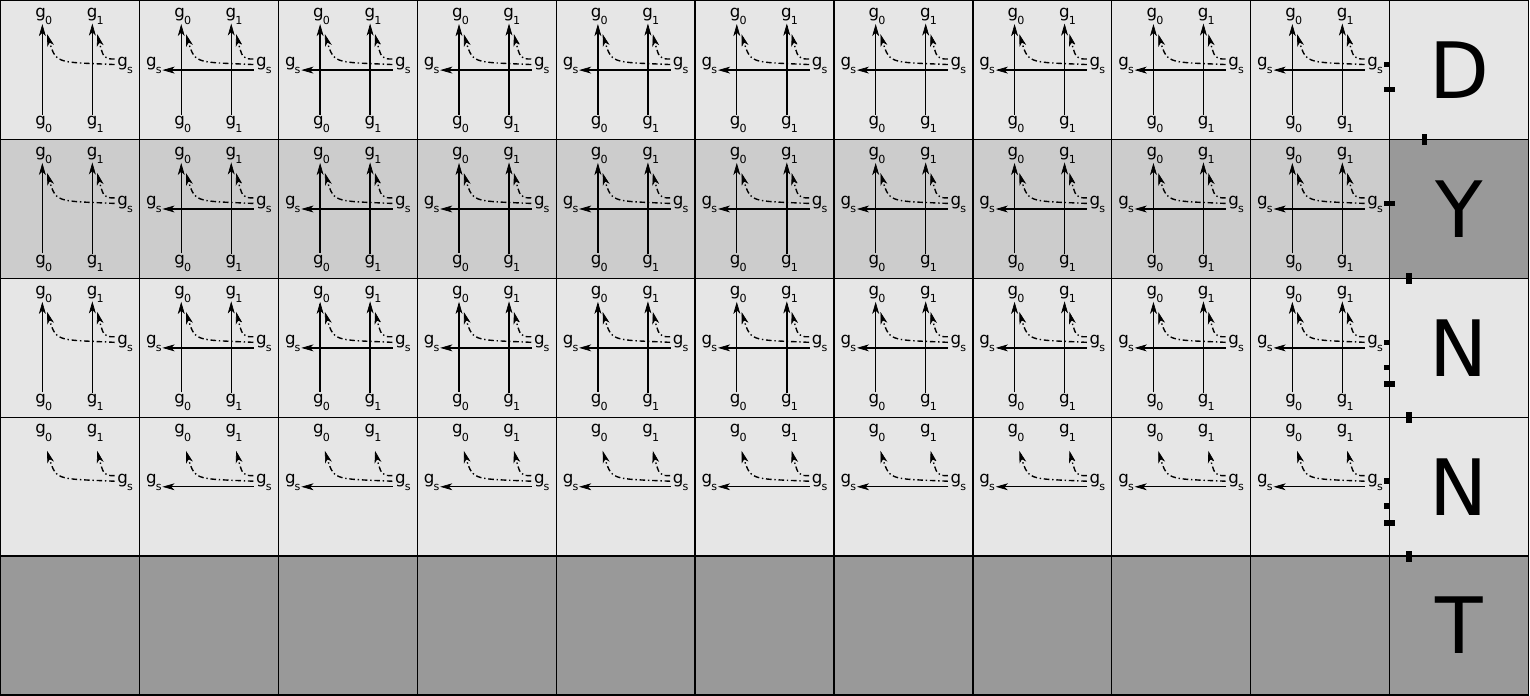}
\caption{A schematic figure for a seed. For each tile, depending on the encoding of the tile set to be simulated, only one of the signals depicted with dashed lines is present.}
\label{fig:stam+iu_seed_gadget}
\end{center}
\end{figure}

\subsection{An STAM$^+$ Tile Set which is IU for the STAM$^+$ at $\tau = 1$}

The construction for $\tau=1$ is very similar to the previous, essentially including a few relatively minor changes to adapt the construction of Section~\ref{sec:3D-construction-details}. We begin by noting that by Theorem~\ref{thm:fanOutEssential}, since $\tau = 1$ the macrotiles that are simulating tiles in $\mathcal{T}$ must be capable of simulating two signals, since at $\tau=1$, the tiles of $T$ cannot be less than $2$-simplified. We also note that at $\tau =1$, STAM$^+$ systems are capable of simulating temperature-$2$ zig-zag growth, which can be used throughout the construction of Section~\ref{sec:3D-construction-details} to build the macrotiles. Figure~\ref{fig:zig-zag-growth} shows how, starting from a single row of tiles, one additional row of tiles can be added in a zig-zag fashion. Additional rows can then be grown similarly. In Figure~\ref{fig:zig-zag-growth}, note that the zig-zag growth is deterministic. For example, $T_6^\prime$ and $T_2^\prime$ determine a single type of tile to be placed to the north of $T_2^\prime$. In our case, deterministic zig-zag growth is all that is needed.  Thus, additional glue and signal complexity allows each tile to attach with a single strength-$1$ glue, which causes a signal cascade that exposes a glue combining information about both glues adjacent to the next location for a tile to attach.  By combining information about both adjacent glues, this strength-$1$ glue effectively simulates them and thus their cooperative behavior.

\begin{figure}[htb]
	\centering
		\subfloat[][]{\includegraphics[width=1.5in]{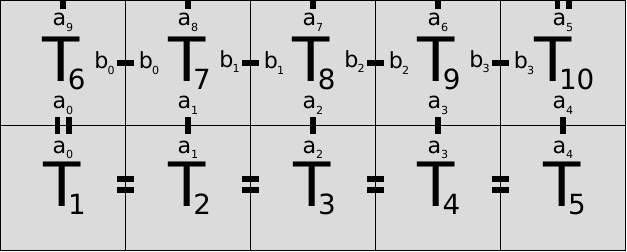}\label{fig:zig-zag-row}}%
		\qquad\qquad\qquad\qquad
		\subfloat[][]{\includegraphics[width=1.5in]{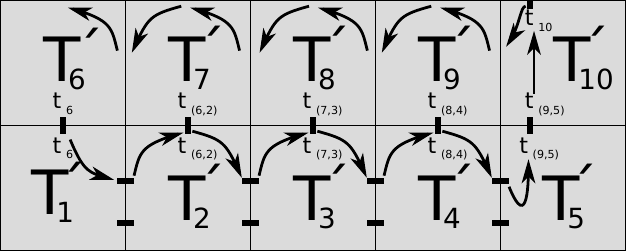}\label{fig:zig-zag-temp1}}
	\caption{ (a) depicts zig-zag growth for $\tau=2$. With tiles $T_1$ through $T_5$ in place, tiles $T_6$ through $T_{10}$ must be placed in that order. Moreover, tile $T_6$ depends only on $T_1$, while tiles $T_i$ for $i \in \{7,8,9,10\}$ depend on $T_{i-1}$ and $T_{i-5}$.
			  (b) depicts zig-zag growth for $\tau=1$ using active tiles. Again, with tiles $T_1^\prime$ through $T_5^\prime$ in place, tiles $T_6^\prime$ through $T_{10}^\prime$ must be placed in that order. Moreover, tile $T_6^\prime$ depends only on $T_1^\prime$, while tiles $T_i^\prime$ for $i \in \{7,8,9,10\}$ depend on $T_{i-1}^\prime$ and $T_{i-5}^\prime$. Also, tiles in each consecutive row contain signals that allow for another row to grow in a zig-zag pattern. Note that only the signals which activate glues which are eventually bound to are shown, while the tiles would actually need to accommodate glues and signals which could represent all possible pairs of input glues for the next growth location.}
\label{fig:zig-zag-growth}%
\end{figure}

\subsubsection{Macrotile formation}

Now, recall that the growth of the macrotiles described in Section~\ref{sec:MacroCreation} that occurs in the $z=0$ plane follows a zig-zag pattern. Therefore, we can use active tiles such that similar growth can occur in 2D STAM$^+$ systems with $\tau = 1$. Figure~\ref{fig:zig-zag-macrotile} shows a scheme for how this zig-zag growth works. Starting from a seed, a tile type $t \in T$ is selected. This seed and selection process is described in Section~\ref{sec:stam+seed} that follows. Then, following a zig-zag growth pattern, each side of the macrotile forms. As the sides of the macrotiles form, special care is taken in presenting the glues exposed by each edge of the macrotile. The glues that will be exposed on the exterior of a macrotile are $\texttt{latent}$ initially. This ensures proper growth of the macrotile
by postponing macrotile binding until the macrotile has formed. (See figure~\ref{fig:improperGrowth}.) As zig-zag growth forming the west edge of a macrotile completes, a glue is exposed on the west edge of the northwest most tile, allowing for the attachment of a \emph{corner gadget} at the northwest corner. When the northwest corner gadget binds, it initiates the successive firing of signals that eventually expose a glue that allows for the growth of the north side of the macrotile and then a second corner gadget to attach at the northeast corner. Similarly, when the northeast corner gadget binds, it initiates the successive firing of signals that eventually expose a glue that allows for a third corner gadget to attach at the southeast corner and in turn fires a signal that eventually allows for the final corner gadget to be placed at the southwest corner. Just as in the construction of Section~\ref{sec:MacroCreation}, the corner gadgets give the macrotile bumps which ensure proper alignment when two macrotiles bind. Finally, the binding of the final corner gadget initiates a circuit of active glues that trigger the simulated glues exposed by the macrotile to turn $\texttt{on}$. Attaching corner gadget in this manner guarantees us that the bumps of the macrotile are in place prior to the exposure of any $\texttt{on}$ glues on the exterior of the macrotile.

\begin{figure}[htp]
\begin{center}
\includegraphics[width=4in]{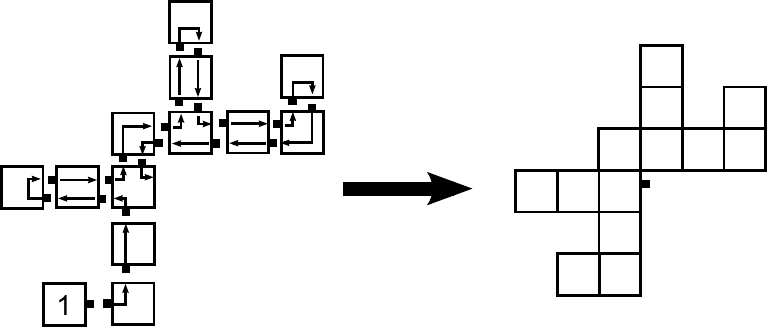}
\caption{A corner gadget for macrotiles in $\tau = 1$ systems. Note that the formation of the entire corner gadget must take place prior to the exposure of the glue that allows the corner gadget to bind to a macrotile.}
\label{fig:corner-gadget}
\end{center}
\end{figure}

\begin{figure}[htp]
\begin{center}
\includegraphics[width=4in]{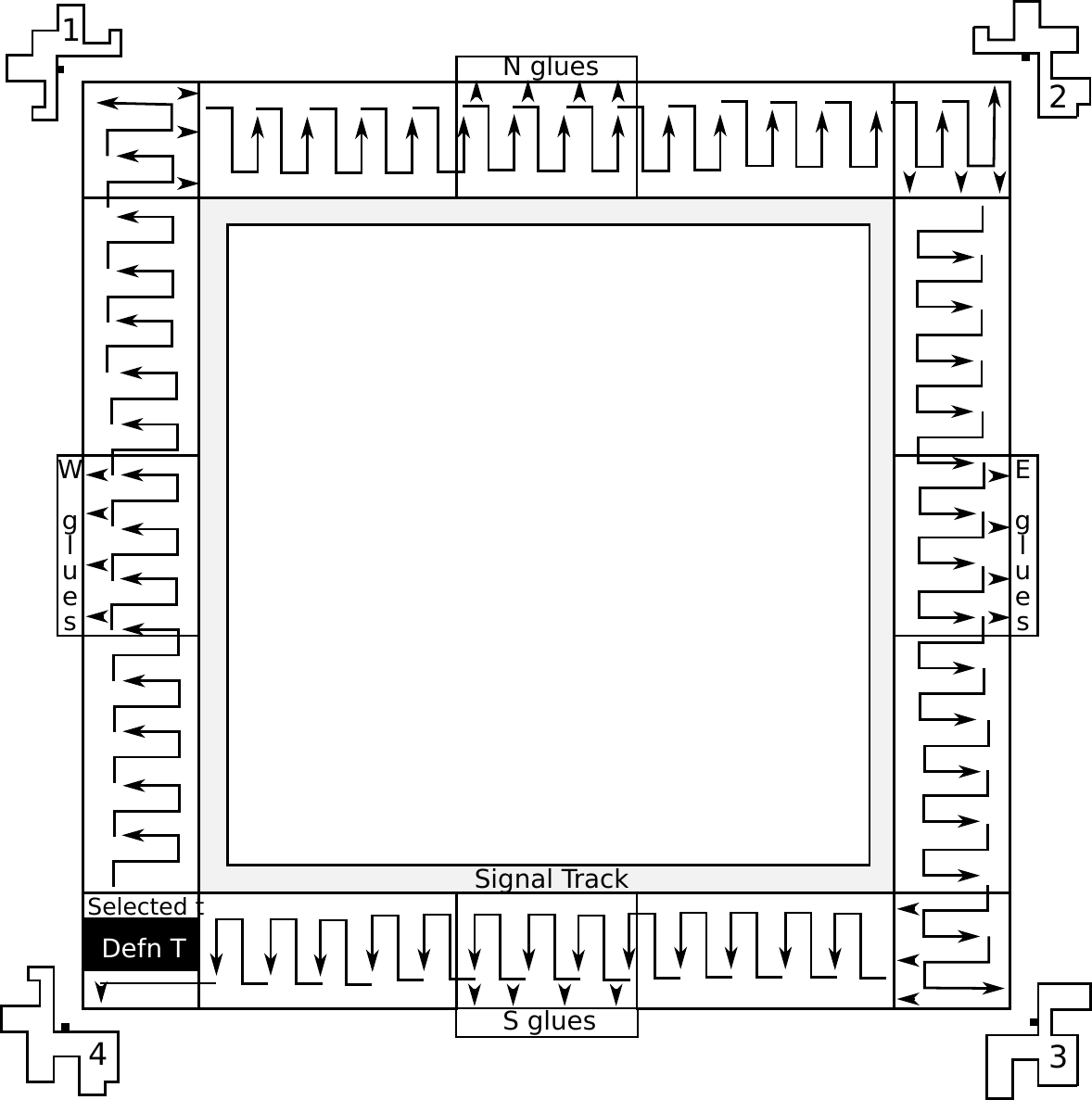}
\caption{A schematic figure for a macrotile for systems with $\tau = 1$. Corner gadgets are labeled in the order that they must attach.}
\label{fig:zig-zag-macrotile}
\end{center}
\end{figure}

\subsubsection{Handling signals}

To permit the firing of signals by glues on the exterior of a macrotile, we proceed as in the $\tau > 1$ case. Except here, the formation of the signal track takes place during the zig-zag formation of a macrotile $\alpha$.  Again, note that for any macrotile, we need only fire at most $2$ signals. We define a {\em wire} to be a path of active glues through tiles of a macrotile such that the signals along this path propagate
through adjacent tiles starting at a preset initial tile until the signal reaches a preset ending tile. For the $\tau=1$ construction we will turn glues on macrotile edges $\texttt{on}$ using wires. A binding event involving a glue $g_1$ can initiate the propagation of signals that turn $\texttt{on}$ a glue $g_2$. In this case, we say that there is a \emph{wire} from $g_1$ to $g_2$.
Now, in the case where $\tau = 1$, glues that allow for two macrotiles to bind are $\texttt{latent}$ initially.
Let $\alpha$ be a macrotile producible in $\mathcal{U}_{\mathcal{T}}$, let $a$ be a glue of the macrotile that fires a signal and let $b$ be a glue that turns $\texttt{on}$ as a result of the binding of $a$. As alpha grows, wires can be placed on tiles using the information encoded in the seed by simply counting to the correct glue locations. These wires start from the tile with the $a$ exposed and end at the tile with a glue $b$ glue exposed. The placement of tiles with the signals making up the wire from $a$ to $b$ is as follows. First, if $a$ is exposed on a tile, $t_a$, belonging to the east or west edge (resp. the north or south edge) of $\alpha$, tiles containing signals making up the wire from $a$ to $b$ are place horizontally (resp. vertically) as a row (resp. column) of the zig-zag pattern from $t_a$ to a tile $t_s$ lying on the signal track. The signal track is depicted in figure~\ref{fig:zig-zag-macrotile}.
Similarly, if $b$ is a glue on a tile, $t_b$, belonging to the east or west edge (resp. the north or south edge) of $\alpha$, tiles containing signals making up the wire from $a$ to $b$ are place horizontally (resp. vertically) as a row (resp. column) of the zig-zag pattern from $T_b$ to a tile $t_e$ lying on the signal track.
The final tiles containing signals making up the wire from $a$ to $b$ are the tiles along the signal track that run clockwise from $t_s$ to a tile $t_e$.
Notice that we can have two such wires from any two glues exposed on the exterior of $\alpha$.
Now, when $\alpha$ binds to another macrotile $\beta$ producible in $\mathcal{U}_{\mathcal{T}}$ using $a$, the binding event for $a$ initiates the propagation of signals that eventually expose a $b$ via a wire from $a$ to $b$.

\subsubsection{The initial configuration for $\tau = 1$ }\label{sec:stam+seed}

In this section, we show how to modify the seed given in Section~\ref{sec:stam+seed_temp2} so that it can be used in systems where $\tau=1$. The main idea here is to use signals to mimic the behavior of cooperation that can occur in systems with temperature greater than $1$. Figure~\ref{fig:stam+seedTemp1} gives an example of such a seed supertile. The rows of the supertile that encode tiles of $T$ and the selection process remain unchanged, however the nondeterministic selection of a row changes slightly. The binding of successive selector tiles (tiles labeled $N$, $Y$ and $D$) use signals to ensure that they are placed from bottom to top. Also, a stopper tile, labeled $S$ in the figure, prevents the placement of an extra selector tile. Finally, as in the case for $\tau>1$, we must ensure that some row is selected. To do this we place a glue $d$ on the south edge of tile $S$ and a $d$ glue on the north edge of each tile labeled $N$ so that if a tile labeled $N$ is placed below $S$, binding of $d$ with fire a signal to turn $\texttt{on}$ $g_s$.

\begin{figure}[htp]
\begin{center}
\includegraphics[width=4in]{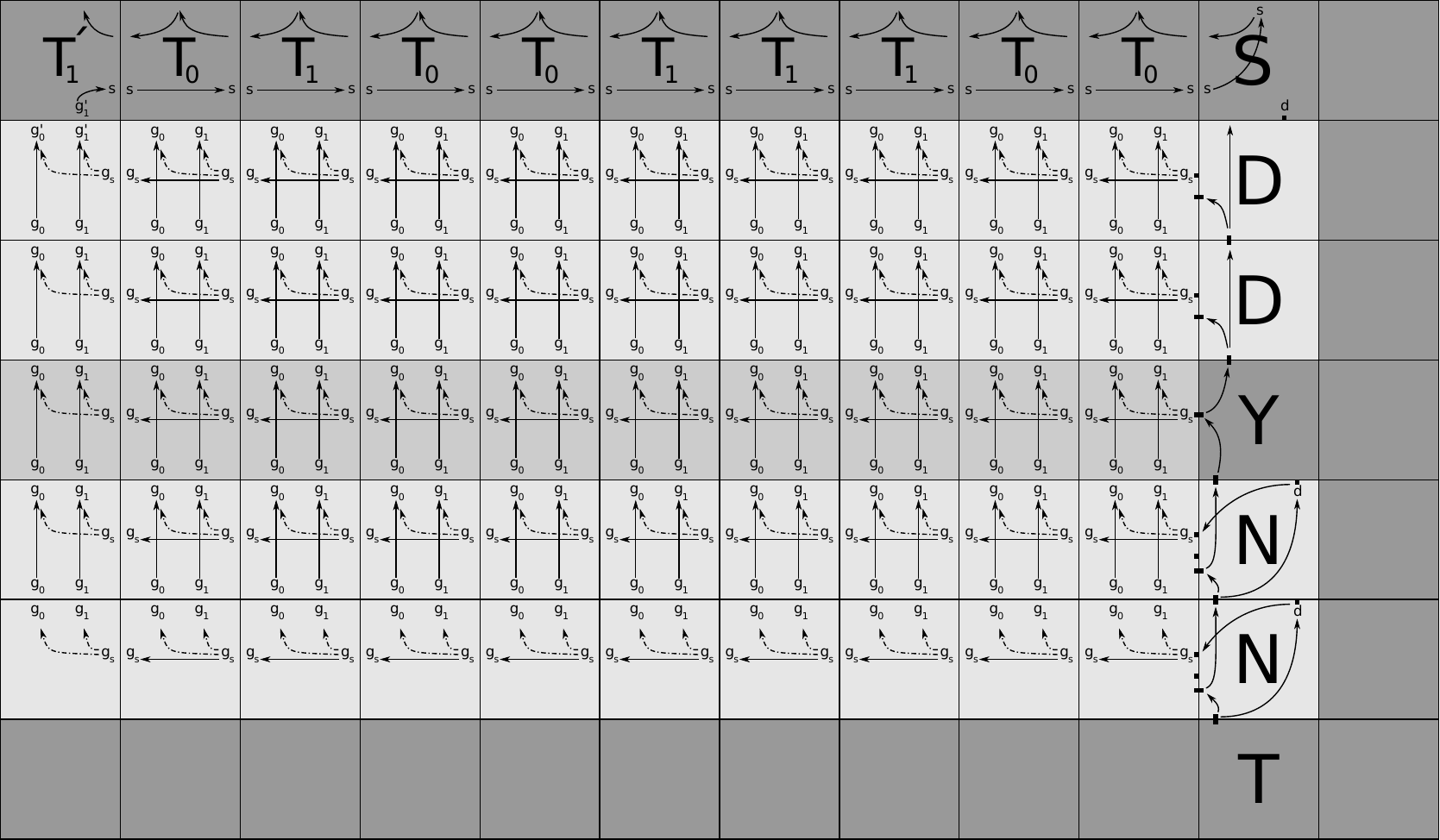}
\caption{A schematic figure for a seed for $\tau = 1$. For each tile, depending on the encoding of the tile set to be simulated, only one of the signals depicted with dashed lines is present. Tiles labeled $T_0$, $T_1$, $T_1^\prime$, $N$, $Y$ and $D$ are not part of the seed assembly and are only shown to illustrate the selection of a row.}
\label{fig:stam+seedTemp1}
\end{center}
\end{figure}

Once a row is selected, glues that have been turned $\texttt{on}$ during the selection process allow for tiles, labeled $T_0$, $T_1$ and $T_1^\prime$ in the figure, to be placed. A tile $T_0^\prime$ analogous to the tile $T_1^\prime$ is place instead of $T_1^\prime$ if the left most bit of the selected row is $0$. $T_0^\prime$ is not shown in the figure. When $T_1^\prime$ (or $T_0^\prime$) binds, signal propagated by glue pair bindings, corresponding to glues labeled $s$ in the figure, are used to ensure that each $T_0$, $T_1$, $T_0^\prime$ or $T_1^\prime$ is in place. Once each of these tiles is in place, a glue is presented on the north edge of the rectangular supertile grown from the seed that allows for zig-zag growth of the macrotile to start.

Finally, we make a general observation about the seed supertile formation for systems with $\tau = 1$. Consider the glues that hold the initial configuration together. Such glues cannot be exposed initially, for if they were, tiles that make up the initial configuration could spontaneously bind and possibly create arbitrary seeds containing encodings of tile types that should not be simulated. ``Junk'' supertiles that do not take part in any simulation could also be created. Therefore, for systems with $\tau=1$ we form the seed from tiles which initially have all glues in the $\texttt{latent}$ state, and first manually turn $\texttt{on}$ the necessary set of glues on each tile used for the seed, and then use these glues to attach those tiles together.

\subsection{Complexity analysis}

First, the tile complexity is $O(1)$ as there is a single universal tile set $U_{\tau}$ which simulates any $T'$. Let $k$ denote the tile complexity in the case where $\tau>1$.
Now, since macrotiles form in the same way as the $z=0$ plane growth of the 3D macrotiles given in Section~\ref{sec:3D-construction-details}, the scale factor of the macrotiles in the STAM$^+$ construction the same. In particular, the scale factor is $O(|T|\log|T|)$, while using a constant sized tile set and the scale factor for $\mathcal{U}_{\mathcal{T}}$ simulating $\mathcal{T}'$ is $O(|T'|^4 \log(|T'|^2))$. The techniques to reduce signal complexity given in Section~\ref{sec:transform-to-simple-stam} can be used to reduce the signal complexity of tiles in our macrotiles. We claim that this only contributes a constant size increase in scale-factor. First, note that the most signal complexity for any one tile of the construction at $\tau > 1$ is $4$, and the most glue complexity of any one tile is bounded by a constant that depends only on $k$. In particular, the glue complexity is independent of $|T'|$. Hence, the tile set $U_{\tau}$ can be $1$-simplified with only a constant sized increase in scale-factor. Therefore, the scale factor for $\mathcal{U}_{\mathcal{T}}$ simulating $\mathcal{T}'$ is $O(|T'|^4 \log(|T'|^2))$ even for a $1$ simplified tile set.
For $\tau=1$, the most signal complexity for any one tile of the construction at $\tau = 1$ is also a constant that is independent of the tile set $T$.  In fact, to mimic zig-zag growth for $\tau=1$ each tile contains a number of signals that depends on $k$ plus some signals that are used for the firing of signals once macrotiles have formed.  Overall, this will be a constant amount of signal complexity that is independent of $|T|$.  The glue complexity of each tile is also bounded by a similar constant that is independent of $|T|$. Therefore, when we apply the simplification techniques given in Section~\ref{sec:transform-to-simple-stam}, the scale-factor increase will be a constant independent of $|T|$. Hence, the tile set $U_{\tau}$ can be $2$-simplified with only a constant sized increase in scale-factor, and the scale factor for $\mathcal{U}_{\mathcal{T}}$ simulating $\mathcal{T}'$ remains to be $O(|T'|^4 \log(|T'|^2))$ for $\tau=1$.

\subsection{Representation function and correctness of simulation}
The proof of this is similar to the proof in Section~\ref{sec:3D2HAMIUProof}.

\section{Conclusion}

We have shown how to transform STAM$^+$ systems (at temperature $1$ or $>1$) of arbitrary signal complexity into STAM$^+$ systems which simulate them while having signal complexity no greater than $2$ and $1$, respectively.  However, if the original tile set being simulated is $T$, the scale factor and tile complexity of the simulating system are approximately $O(|T|^2)$.  It seems that these factors cannot be reduced in the worst case, i.e. when a tile of $T$ has a copy of every glue of the tile set on each side, and each copy of each glue on the tile activates every other, yielding a signal complexity of $O(|T|^2)$.  However, whether or not this is a true lower bound remains open, as well as what factors can be achieved for more ``typical'' systems with much lower signal complexity.

A significant open problem which remains is that of generalizing both constructions (the signal reduction and the 3D 2HAM simulation) to the unrestricted STAM.  Essentially, this means correctly handling glue deactivation and possible subassembly dissociation.  While this can't be handled within the standard 3D 2HAM where glue bonds never change or break, it could perhaps be possible if negative strength (i.e. repulsive) glues are allowed (see \cite{DotKarMasNegativeJournal} for a discussion of various formulations of models with negative strength glues).  However, it appears that since both constructions use scaled up macrotiles to represent individual tiles of the systems being simulated, there is a fundamental barrier.  The STAM assumes that whenever two tiles are adjacent, all pairs of matching glues across the adjacent edge which are both currently \texttt{on} will immediately bind (which is in contrast to other aspects of the model, which are asynchronous).  Since both constructions trade the ability of individual tile edges in the STAM to have multiple glues with scaled up macrotiles which distribute those glues across individual tiles of the macrotile edges, it appears to be difficult if not impossible to maintain the correct simulation dynamics.  Basically, a partially formed side of a macrotile could have only a subset of its initially \texttt{on} glues in place, but enough to allow it to bind to another macrotile.  At that point, if glue deactivations are initiated which result in the dissociation of the macrotile before the remaining glues of the incomplete macrotile side assemble, then in the simulating system, those additional glues won't ever bind.  However, in the simulated system they would have.  This results in a situation where, after the dissociation, the simulated system would potentially have additional pending glue actions (initiated by the bindings of the additional glues) which the simulating system would not, breaking the simulation.

Overall, laboratory experiments continue to show the plausibility of physically implementing signalling tiles \cite{Jennifer}, while previous theoretical work \cite{Signals} shows some of their potential, and the results in this paper demonstrate how to obtain much of that power with simplified tiles.  We feel that research into self-assembly with active components has a huge amount of potential for future development, and continued studies into the various tradeoffs (i.e. complexity of components, number of unique component types, scale factor, etc.) between related models provide important context for such research.  We hope that our results help to contribute to continued advances in both theoretical and experimental work along these lines.

\iffull
\fi
\bibliographystyle{abbrv} %
\bibliography{tam,experimental_refs,ca}

\newpage

\ifabstract
\else
\appendix
\magicappendix
\fi

\end{document}